\documentclass[11pt]{article}
\usepackage[utf8]{inputenc}
\usepackage[a4paper,width=178mm,top=22mm,bottom=15mm,bindingoffset=6mm]{geometry}
\usepackage{etex}
\usepackage[all]{xy} 
\usepackage{etoolbox} 
\usepackage{lmodern}
\usepackage{array}
\usepackage{xstring}
\usepackage{longtable}
\usepackage[listings]{tcolorbox} 
\tcbuselibrary{breakable}
\tcbuselibrary{skins}
\usepackage[pdftex]{hyperref}
\usepackage{amsmath}
\usepackage{tikz,tikz-cd,stmaryrd}
\usepackage{fancyhdr}
\usepackage{amssymb}
\usepackage{eqnarray,amsmath}
\usepackage{amsthm}
\usepackage{amsmath} 
\usepackage{amsfonts}
\usepackage{amsthm}
\usepackage{pifont}
\usepackage{amsmath}
\usepackage{epsfig}
\usepackage{latexsym}
\usepackage{amssymb}
\usepackage{color}
\usepackage{amscd}
\usepackage{multirow}
\usepackage{graphicx}
\usepackage{lscape}
\usepackage{float}
\usepackage{bm}
\usepackage{soul}
\usepackage{tikz-cd}
\newcommand{\oplustilde}{\mathbin{\tilde{\oplus}}}
 \makeatletter
    \let\stdchapter\section
    \renewcommand*\section{%
    \@ifstar{\starchapter}{\@dblarg\nostarchapter}}
    \newcommand*\starchapter[1]{%
        \stdchapter*{#1}
        \thispagestyle{fancy}
        \markboth{\MakeUppercase{#1}}{}
    }
    \def\nostarchapter[#1]#2{%
        \stdchapter[{#1}]{#2}
        \thispagestyle{fancy}
    }

\makeatletter 
\@addtoreset{equation}{section} \@addtoreset{equation}{subsection} 
\makeatother

\newtheorem{theorem}{Theorem}[section]
\newtheorem*{theorem*}{Theorem}
\newtheorem{lemma}[theorem]{Lemma}
\newtheorem{proposition}[theorem]{Proposition}
\theoremstyle{definition}
\newtheorem{definition}[theorem]{Definition}
\theoremstyle{corollary}
\newtheorem{corollary}[theorem]{Corollary}
\theoremstyle{remark}
\newtheorem{remark}[theorem]{Remark}

\theoremstyle{conclusion}

\usepackage{amsmath}
\usepackage{amssymb}
\usepackage{amsfonts}
\usepackage{esvect}
\usepackage{mathrsfs}
\usepackage{authblk}
\usepackage{geometry}
\usepackage{abstract}
\usepackage{xcolor}

\title{\bf Algebraic structures and Hamiltonians from the equivalence classes of 2D conformal algebras}
\author{\large Ian Marquette \footnote{I.Marquette@latrobe.edu.au}, 
Junze Zhang \footnote{junze.zhang@uqconnect.edu.au, Corresponding author} and Yao-Zhong Zhang \footnote{yzz@maths.uq.edu.au} }
\affil{School of Mathematics and Physics, The University of Queensland \\ Brisbane, QLD 4072, Australia}
\affil{Department of Mathematical and Physical Sciences, La Trobe University \\Bendigo, VIC 3552, Australia}
\begin{document}

\maketitle

\begin{abstract}
\noindent The construction of superintegrable systems based on Lie algebras and their universal enveloping algebras has been widely studied over the past decades. However, most constructions rely on explicit differential operator realisations and Marsden-Weinstein reductions. In this paper, we develop an algebraic approach based on the subalgebras of the 2D conformal algebra $\mathfrak{c}(2)$. This allows us to classify the centralisers of the enveloping algebra of the conformal algebra and construct the corresponding Hamiltonians with integrals in algebraic form. It is found that the symmetry algebras underlying these algebraic Hamiltonians are six-dimensional quadratic algebras. The Berezin brackets and commutation relations of the quadratic algebraic structures are closed without relying on explicit realisations or representations. We also give the Casimir invariants of the symmetry algebras.
 Our approach provides algebraic perspectives for the recent work by Fordy and Huang on the construction of superintegrable systems in the Darboux spaces. 
\end{abstract}

\section{Introduction}
Superintegrable systems have been studied from different perspectives over the years, encompassing approaches from classical and quantum mechanics, representation theory, and mathematical physics \cite{MR3119484,MR2804560,MR3863299,MR4071113,MR4584218,MR4644061}.  In classical mechanics, a system with a Hamiltonian $\mathcal{H}$ is superintegrable in a phase space with dimension $2n$ if there are at least $2n-k$ globally functional independent integrals of motion that Poisson commutes with $\mathcal{H}$, where $1 \leq k \leq n-1$. In quantum mechanics, superintegrability implies the presence of $2n-k$ quantum observables that commute with the quadratic operator $\hat{\mathcal{H}}$.  It has been found that symmetry algebra approaches can be used to analyse physical systems more effectively than other approaches \cite{MR2804560,post2024racah,kuru2023quantum,MR1814439,MR2337668,MR2226333}. It was shown how Lie algebras, their realisations, and reductions allow one to obtain superintegrable Hamiltonians in the context of classical mechanics \cite{Marsden1974,Miller1981,MR0460751}. For instance, 
in \cite{Calzada2006,delOlmo1993,MR1170508,MR372124} various superintegrable systems on spheres and pseudo-spheres have been obtained. However, 
most of the previous studies rely on the use of explicit differential operator realisations, the Marsden-Weinstein reductions, and connections with special functions \cite{MR3702572,MR2384724,MR2023556}.

Only very recently it was found that an entirely algebraic approach to superintegrable systems can be established by realizing that the integrals of the systems are as close to polynomial algebras in the enveloping algebras of certain Lie algebras. The approach was successfully implemented for $\mathfrak{su}(3)$ \cite{Correa2020} and $ \mathfrak{gl}(3)$ \cite{CampoamorStursberg2021} by applying a scheme based on commutants and algebraizations. It is also extended to commutants regarding the Cartan subalgebras of $\mathfrak{sl}(n)$ \cite{CampoamorStursberg2023} and certain non-semisimple algebras \cite{CampoamorStursberg2022}. In this paper, we generalise the procedure further and present the construction of commutants concerning subalgebras, the corresponding algebraic superintegrable Hamiltonians, and their integrals of motion. This allows us to propose a new scheme for classifying superintegrable systems by classifying subalgebras and related commutants of a given Lie algebra.

Recently it was realised that conformal algebras play an important role in classifying superintegrable systems 
in 2-dimensional spaces \cite{MR2804560,MR2143019,MR2143027,MR3116192}. In the work by Fordy and Huang \cite{MR3988021,MR4088503,MR4258337}, through the use of specific subalgebras, their Casimir operators and differential realisations, it was shown that superintegrable systems in $2$D and $3$D conformally flat spaces are related to superintegrable systems in Darboux spaces.
In particular, the conformal symmetries of the $2$D Euclidean metric were applied to construct new free superintegrable systems in the 2D Darboux spaces. In this paper, we will examine various Abelian and non-Abelian subalgebras of the 2D conformal algebra $\mathfrak{c}(2)$ and propose a new scheme for classifying integrable and superintegrable systems through explicit construction of commutants relative to the subalgebras. 

Our construction is independent of the metric of the space. Using polynomials up to degree $6$ in the enveloping algebra of $\mathfrak{c}(2)$,  we obtain the corresponding algebraic Hamiltonians, integrals, and polynomial symmetry algebras.
Here, we use only the $Ad(C(2))$-equivalence to each subalgebra of $\mathfrak{c}(2)$ to build the polynomial algebras, where is $C(2)$ the corresponding conformal group of $\mathfrak{c}(2).$ As an example of the systematic scheme, we look at one type of commutant involving the Casimir of a subalgebra and obtain the corresponding algebraic superintegrable Hamiltonian. Our approach establishes a framework linking superintegrable systems to the theory of Lie algebras, their subalgebras, and enveloping algebras.  

The structure of the paper is as follows. In Section $\ref{2}$, we introduce a general approach for constructing polynomial algebras from the centralisers of subalgebras $\mathfrak{a}$ of a finite-dimensional Lie algebra $\mathfrak{g}$ in the universal enveloping algebra $\mathcal{U}(\mathfrak{g})$, and the corresponding algebraic superintegrable systems based on these reduction chains.   Then in Section $\ref{3},$ we focus on the 2D conformal algebra $\mathfrak{c}(2)$. Since $\mathfrak{c}(2)$ is an algebra isomorphic to the simple Lie algebra $\mathfrak{so}(3,1),$ we will use the classification of $\mathfrak{so}(3,1)$ given in \cite{MR372124} to identify the subalgebras of $\mathfrak{c}(2)$. Section \ref{4} shows that a new polynomial algebra arises from the centraliser of each subalgebra in the universal enveloping algebra of $\mathfrak{c}(2)$. Similarly, we determine finitely-generated Poisson algebras hidden in the symmetric algebra $\mathcal{S}(\mathfrak{c}(2)).$ Moreover, it turns out that the non-Abelian centralisers can be realised as certain quadratic algebras. Moreover, the Casimir invariants are also obtained. Under an appropriate realisation, some of the quadratic algebras can be identified with the Racah-type algebras. To relate to the work by Fordy and Huang in the $2$D conformal space, we take the Casimir $K_\mathfrak{a}$ of each subalgebra of $\mathfrak{c}(2)$ as the Hamiltonian. As can be seen in Section $\ref{4.5}$, the symmetry algebras generated by the centralisers $\mathcal{C}_{\mathcal{U}(\mathfrak{g})}(K_\mathfrak{a})$ are different from those of the centralisers $\mathcal{C}_{\mathcal{U}(\mathfrak{g})}( \mathfrak{a})$. In Section $\ref{5}$, we discuss the relationship between the quadratic algebras defined in Subsection $\ref{4.1}$ and the general quadratic algebras defined in \cite{Marquette:2023wxn}. Finally, Section $\ref{6}$ provides the conclusions of the paper.

\section{Superintegrable systems from the centralizers of $\mathcal{U}(\mathfrak{g})$}
\label{2}

We begin with foundational terms in the Poisson-Lie framework; see \cite[Chapter 7]{MR2906391} for further details. Let $ \mathfrak{g} $ be a $n$-dimensional Lie algebra over a characteristic zero field $\mathbb{F}$ with an ordered basis $\beta_\mathfrak{g} = \big(X_1,\ldots,X_n\big)$ satisfying the commutator relations $[X_i,X_j]_o = \sum_{k = 1}^n C_{ij}^k X_k$, and let $ \mathfrak{g}^* $ be its dual Lie algebra.  For any $f,g \in C^\infty(\mathfrak{g}^*)$ and $\xi \in \mathfrak{g}^*$, one can define the Poisson-Lie bracket in $\mathfrak{g}^*$ at $\xi$ as follows: \begin{align}
    \{f,g\}_o(\xi)  = \langle \xi, [d_\xi f,d_\xi g]_o \rangle ,\label{eq:poiss}
\end{align} where $d_\xi f, d_\xi g \in  T_\xi^* \mathfrak{g}^* \cong \left(\mathfrak{g}^*\right)^*  \cong \mathfrak{g}$, and $\langle \cdot,\cdot \rangle$ is the dual pair between $\mathfrak{g}^*$ and $\mathfrak{g}$. We now consider how this structure works in coordinates of $\mathfrak{g}^*$. Define the corresponding  coordinate functions $  \big(x_1,\ldots,x_n\big)$ on $\mathfrak{g}^*$ by $x_i(\xi) = \langle \xi,X_i\rangle$ with $i = 1,\ldots,n$. For any coordinate functions $x_i,x_j $, using the general formulas for the Poisson-Lie brackets \eqref{eq:poiss}, one can show that the Poisson bracket of the coordinate functions is given by \begin{align}
    \{x_i,x_j\}_o = \sum_{k=1}^n C_{ij}^k x_k.
\end{align}   In this case, the restriction of the Poisson bracket to the coordinates functions $x_j$ ($1 \leq j \leq n$) coincides with the Lie brackets on $\mathfrak{g}$.

Let $\mathcal{U}(\mathfrak{g})$ and $\mathcal{S}(\mathfrak{g})$ be the universal enveloping algebra and the symmetric algebra of $\mathfrak{g}$ , respectively. From the universal property, $\mathcal{U}(\mathfrak{g})$ is closed under the commutator $[\cdot,\cdot]_o$. On the other hand, since $\mathcal{S}(\mathfrak{g}) \subset C^\infty(\mathfrak{g}^*)$ is a Poisson subalgebra, the Poisson-Lie bracket defined in \eqref{eq:poiss} also admits on $\mathcal{S}(\mathfrak{g})$. Using the universal property of $\mathcal{S}(\mathfrak{g})$, we identify $\mathcal{S}(\mathfrak{g})$ as the polynomial algebra $\mathbb{F}[\mathfrak{g}^*] := \mathbb{F}[x_1,\ldots,x_n]$ through an algebra isomorphism, where the coordinate functions $x_j : \mathfrak{g}^* \rightarrow \mathbb{R} $ are realised as the generators of $\mathfrak{g}$. In what follows, we always refer the symmetric algebra $\mathcal{S}(\mathfrak{g})$ to the polynomial algebra $\mathbb{F}[x_1,\ldots,x_n]$.

We now define the commutants relative to the subalgebras of $\mathcal{U}(\mathfrak{g})$ and $\mathcal{S}(\mathfrak{g})$.  

\begin{definition}
\label{2.1}
 The $\textit{commutants}$ relative to the subalgebras $ \mathfrak{a} \subset \mathfrak{g}$ and $\mathfrak{a}^* \subset \mathfrak{g}^*$ in $\mathcal{U}(\mathfrak{g})$ and $\mathcal{S}(\mathfrak{g})$, denoted respectively as $\mathcal{C}_{\mathcal{U}(\mathfrak{g})}(\mathfrak{a})$   and $\mathcal{C}_{\mathcal{S}(\mathfrak{g})}(\mathfrak{a})$, are defined as the centralisers of $\mathfrak{a} $ in $\mathcal{U}(\mathfrak{g})$ and $\mathcal{S}(\mathfrak{g})$,
\begin{align*}
      \mathcal{C}_{\mathcal{U}(\mathfrak{g})}(\mathfrak{a}) &= \left\{ Y \in \mathcal{U}(\mathfrak{g}): [X,Y]_o = 0  \quad \forall X \in \mathfrak{a}\right\},\\
      \mathcal{C}_{\mathcal{S}(\mathfrak{g})}(\mathfrak{a})  & = \{y \in \mathcal{S}(\mathfrak{g}): \{x,y\}_o = 0  \quad \forall x \in \mathfrak{a}^*\},
  \end{align*} where $[\cdot,\cdot]_o$ is the usual Lie bracket that satisfies $[X_i,X_j]_o = C_{ij}^k X_k$ and $\{\cdot,\cdot\}_o$ is the corresponding the canonical Poisson-Lie bracket in $\mathfrak{g}^*.$
\end{definition} 

\begin{remark}
\label{re2.1}
(i) For any $h \in  \mathbb{N}_0 :=\mathbb{N}\cup \{0\},$ we can define $\mathcal{U}_h(\mathfrak{g}) = \mathrm{span} \left\{X_1^{i_1} \cdots X_n^{i_n}: i_1+ \cdots + i_n \leq  h\right\}$ as the subspace of $\mathcal{U}(\mathfrak{g}).$ We can then define the degree $\delta$ of an arbitrary element $P \in \mathcal{U}(\mathfrak{g})$ as $\delta = \mathrm{inf} \{ h: P\in\mathcal{U}_h(\mathfrak{g})\}.$ Furthermore, there is a natural filtration in $\mathcal{U}(\mathfrak{g}) = \bigcup_{h \in \mathbb{N}_0} \mathcal{U}_h(\mathfrak{g}) $ such that   \begin{align}
    \mathcal{U}_{0}(\mathfrak{g}) = \mathbb{F}, \quad \mathcal{U}_h(\mathfrak{g})\mathcal{U}_k(\mathfrak{g}) \subset \mathcal{U}_{h+k}(\mathfrak{g}), \quad \mathcal{U}_h(\mathfrak{g}) \subset \mathcal{U}_{h+k}(\mathfrak{g}),\quad \forall k \geq 1. \label{eq:fil}
\end{align}

(ii)  Notice that a linear basis in $  \mathcal{C}_{\mathcal{U}(\mathfrak{g})}(\mathfrak{a})$ is not necessarily canonical. In other words, there is no guarantee for the existence of finitely generated Poisson algebras associated with a (arbitrary) finite-dimensional Lie algebra. It was shown that in \cite[Section 2.4.13]{MR1451138}, if $\mathfrak{g}$ is a semisimple or reductive Lie algebra, then the commutant $\mathcal{C}_{\mathcal{U}(\mathfrak{g})}(\mathfrak{a})$ is a finitely generated module on a polynomial ring whenever the subalgebra $\mathfrak{a}$ is reductive in $\mathfrak{g}$. Note also that a linear basis of $\mathcal{C}_{\mathcal{U}(\mathfrak{g})}(\mathfrak{a})$ is not necessarily algebraically independent, but only linearly independent. Moreover, from the calculation in Section \ref{4}, we can see that the algebraic structures of centraliser subalgebras in Definition $\ref{2.1}$ are heavily dependent on the commutator relations of $\mathfrak{g}.$ 

\end{remark}

Note that $\mathcal{S}(\mathfrak{g}) \cong \mathbb{F}[\mathfrak{g}^*]$ as mentioned above and so in the following  each $p \in \mathcal{S}(\mathfrak{g})$ is expressed as a polynomial in terms of $x_1,\ldots,x_n$. 

\begin{definition}
\label{2.3} \cite{MR0760556}
The adjoint action of  $\mathfrak{g}$ and co-adjoint action of  $\mathfrak{g}^*$ on the universal enveloping algebra $\mathcal{U}(\mathfrak{g})$ and the symmetric algebra $\mathcal{S}(\mathfrak{g})$ are given, respectively,  by \begin{align}
    ad(X_j): \quad & P\left( X_1,\ldots,X_n\right) \in \mathcal{U}(\mathfrak{g})  \mapsto  [X_j,P]_o \in \mathcal{U}(\mathfrak{g}) , \\
    ad^*(X_j): \quad &  p(x_1,\ldots,x_n) \in \mathcal{S}(\mathfrak{g})   \mapsto \{x_j,p  \}_o  = \tilde{X}_j(p) = \sum_{l,k = 1}^n C_{jk}^lx_l \dfrac{\partial p}{\partial x_k} \in \mathcal{S}(\mathfrak{g}), \label{eq:dual}
  \end{align} 
where $\tilde{X}_j = \sum_{l,k = 1}^n C_{jk}^lx_l \dfrac{\partial }{\partial x_k}$ is the canonical representation of $\mathfrak{g}$ using linear vector fields in $\mathfrak{g}^*$.
\end{definition}
\begin{remark}
    
Note that the symmetric algebra $\mathcal{S}(\mathfrak{g})$ is an associative polynomial Poisson subalgebra. It can be shown that the Poisson-Lie bracket as described in \eqref{eq:poiss} takes the following coordinate expression:
  \begin{align}
     \{p,q\}_o  = \sum_{l=1}^n C_{jk}^l x_l \dfrac{\partial p}{\partial x_j} \dfrac{\partial q}{\partial x_k} \quad \forall p,q \in \mathcal{S}(\mathfrak{g}) , \label{eq:Poisson}
 \end{align}  where $ (x_1,\ldots,x_n)$ is the  coordinate functions of $\mathfrak{g}^*$.  For more details, see, for instance, \cite[Chapter 6, Section 6.1]{olver1993applications}.
\end{remark}

We now construct the symmetry algebras from the subalgebras $\mathfrak{a}$ of $\mathfrak{g}$.  Without loss of generality, we denote the generators of $\mathfrak{a}$ by $X_{\ell_1},\ldots,X_{\ell_s}$, where $s = \dim \mathfrak{a}$ and $\{\ell_1,\ldots,\ell_s\}$ is a subset of $\{1,\ldots,n\}$. 
Similar constructions were applied in \cite{MR191995,MR4660510,MR1520346}.  In view of Definition $\ref{2.1}$ and $\eqref{eq:dual},$  the commutant $\mathcal{C}_{\mathcal{S}(\mathfrak{g})}(\mathfrak{a})$ is obtained as homogeneous polynomial solutions of the system of partial differential equations (PDEs in short) \begin{align}
\tilde{X}_j(p_h)(x_1,\ldots,x_n) = \{x_j,p_h\}_o  =   \sum_{l,k = 1}^n C_{jk}^lx_l \dfrac{\partial p_h}{\partial x_k} = 0, \quad  j = \ell_1,\ldots,\ell_s ,\label{eq:func}
\end{align} where  $p_h(x_1,\ldots,x_n)$ are homogeneous polynomials of degree $h \geq 1$ with the generic form\begin{align}
    p_h(x_1,\ldots,x_n) = \sum_{i_1 + \cdots + i_n  =h}  \Gamma_{i_1,\ldots,i_n}\, x_1^{i_1} \cdots x_n^{i_n} \in \mathcal{S}(\mathfrak{g}). \label{eq:ci}
\end{align} Notice that, depending on the structure of the Lie algebra $\mathfrak{g},$  solutions to $\eqref{eq:func}$ may not be polynomials. See for instance, \cite{MR4660510} and the reference therein.  In this work, we will assume that the system $\eqref{eq:func}$ admits an integrity basis formed by polynomials. In order to form a finitely-generated algebra, we need to discard all the decomposable solutions from the solution space of $\eqref{eq:func}$. Here a polynomial $p \in \mathcal{S}(\mathfrak{g}) $ is said to be $\textit{decomposable}$ if there exists a polynomial $p'   \in \mathcal{S}(\mathfrak{g})   $ with $p' \neq p$ such that $p  \equiv 0 \mod p' .$ 

 We now present an algorithm for the construction of linearly independent and indecomposable polynomials in an abstract setting. For similar approaches, see, for example, \cite{CampoamorStursberg2023,MR4660510}.  We start from homogeneous polynomial solutions of degree one of $\eqref{eq:func}$ that belong to the centraliser of $\mathfrak{a}$ in $\mathfrak{g}.$  We denote the set of all linearly independent degree one monomials by 
 $$ \textbf{q}_1 = \left\{p_{1,1},\ldots,p_{1,m_1}\right\}. $$ We then construct quadratic solutions of $\eqref{eq:func}.$ Any singular-plural are not agreed of the elements of the set $\textbf{q}_1$ need to be discarded. In this way, we can construct the set of all linearly independent and indecomposable quadratic homogeneous polynomials \begin{align*}
     \textbf{q}_2 = \left\{p_{2,1},\ldots,p_{2,m_2}\right\}.
 \end{align*}  Repeating the process above, all possible indecomposable homogeneous polynomials up to a certain degree $g $, deduced from $\eqref{eq:func},$ can be obtained by composition  \begin{align}
   \textbf{Q}_g = \bigcup_{h=1}^g\textbf{q}_h , \label{eq:bas}
\end{align} where the set $\textbf{q}_h =\{p_{h,1} ,\ldots,p_{h,m_h} \}$ contains all the indecomposable homogeneous polynomials of degree $h$. That is, in the construction of $\eqref{eq:bas}$, the polynomial solutions that can be written as product of indecomposable homogeneous polynomials of lower degrees must be discarded. In what follows, we will assume that $g$ is the maximum degree of a homogeneous polynomial in the set of all commutants $\textbf{Q}_g$. By maximum degree, we mean that any homogeneous polynomials with degree greater than $g$ are decomposable and are not contained in $\textbf{Q}_g.$ In this way, we are able to determine a finite generating set $\textbf{Q}_g$  containing $m_1 + \cdots + m_g$ number of indecomposable monomials. Let $\textbf{Alg} \langle  \textbf{Q}_g \rangle$ denote the algebra generated by the set $\textbf{Q}_g.$


We now define a Poisson bracket on the finitely-generated algebra $\textbf{Alg} \langle  \textbf{Q}_g \rangle$. For any $p_h \in \textbf{q}_h$ and $p_\ell \in \textbf{q}_\ell$,  define the bilinear map $\{\cdot,\cdot\} :  \textbf{Alg} \langle  \textbf{Q}_g \rangle \times \textbf{Alg} \langle  \textbf{Q}_g \rangle\rightarrow \textbf{Alg} \langle  \textbf{Q}_g \rangle$ by \begin{align}
    \{p_h,p_\ell\} =  \sum_{h_1+ \cdots + h_j \leq \ell + h -1} \Gamma_{m_{h_1} \ldots m_{h_j}}^{h \ell} p_{h_1,m_{h_1}} \cdots p_{h_j,m_{h_j}}, \label{eq:newpoi}
\end{align} where $1\leq h_j\leq g$ for $1\leq j\leq k$ and  $\Gamma_{m_{h_1} \ldots m_{h_j}}^{h \ell}$ are constants. Notice that we can verify that for all the examples provided in the later sections of this paper the Jacobi identity holds for the bilinear map $\eqref{eq:newpoi}$.  This operation gives the polynomial solution set of \eqref{eq:func} a Poisson structure, which is an induced Poisson bracket from the Poisson-Lie bracket defined in $\eqref{eq:Poisson}$. 
 The $\textit{degree of a}$ $\textit{polynomial Poisson algebra}$ is defined as  \begin{align}
    d := \max_{h_1+ \cdots + h_j\leq \ell + h -1} \mathcal{N}\left(p_{h_1,m_{h_1}} \cdots p_{h_j,m_{h_j}}\right), \label{eq:degree}
\end{align}where $\mathcal{N}\left(p_{h_1,m_{h_1}} \cdots p_{h_j,m_{h_j}}\right)$ is the number of $p_{h,m_h}$ appearing in $p_{h_1,m_{h_1}} \cdots p_{h_j,m_{h_j}} $, i.e., in the right hand side of $\eqref{eq:newpoi}$.  The algebra $\left(\textbf{Alg} \langle  \textbf{Q}_g \rangle,\{\cdot,\cdot\}\right)$ is a finitely-generated Poisson algebra of degree $d$. In what follows it will simply denoted as  $\mathcal{Q}(d)$. 

We can show that the centralisers in $ \mathcal{U}(\mathfrak{g})  $ are also finitely generated. There is a well-defined canonical linear isomorphism, see, for instance, \cite{MR432819,MR2515551} \cite[Section 2.12]{gtp}, \begin{align}
      \Lambda: \mathcal{S}(\mathfrak{g}) \rightarrow \mathcal{U}(\mathfrak{g}), \text{ } 
     p\left(x_1,\ldots,x_n\right)   \mapsto P(X_1,\ldots,X_n): = \Lambda\left(p\left(x_1,\ldots,x_n\right)\right)   ,  
 \label{eq:symme}
\end{align} in terms of the canonical basis of $\mathcal{S}(\mathfrak{g})$ is defined by \begin{align}
    \Lambda(x_{i_1} \cdots x_{i_n}) = \frac{1}{n!} \sum_{\sigma \in S_n} X_{i_{\sigma(1)}} \cdots X_{i_{\sigma(n)}}  
\end{align} with $S_n$ being the permutation group  on the set $\{1,\ldots,n\}$,   such that $\Lambda\big( \tilde{X} (p(x_1,\ldots,x_n) \big) = [X,\Lambda(p)]_o$ for any $X \in \mathfrak{g}$. Then for all $k,$ $\Lambda(p_k(x_1,\ldots,x_n)) =P_k(X_1,\ldots,X_n)$ will be non-commutative, functionally independent polynomials in $\mathcal{U}(\mathfrak{g}),$ where $g$ is the maximum degree of the generators. Recall that $\mathcal{S}(\mathfrak{g})$ is naturally a graded algebra \begin{align*}
    \mathcal{S} (\mathfrak{g}) = \bigoplus_{k \geq 0} \mathcal{S}^k(\mathfrak{g}),
\end{align*} where $\mathcal{S}^k(\mathfrak{g}) = \mathrm{span} \left\{x_1^{i_1} \cdots x_n^{i_n}: i_1 + \ldots + i_n =k\right\}$. Here $\mathcal{S}^0(\mathfrak{g}) = \mathbb{F}$. From the natural filtration relation \eqref{eq:fil}, one can define a vector space $\mathcal{U}^k (\mathfrak{g}) := \mathcal{U}_k (\mathfrak{g}) /\mathcal{U}_{k-1} (\mathfrak{g})  $. Note that $\mathcal{U}^k(\mathfrak{g})$ are not subalgebras since the product of two homogeneous elements of degree $k$ and $\ell$ lies in the degree $k+ \ell$. That is, for any $p \in\mathcal{S}^{k}(\mathfrak{g}) $ and $q \in \mathcal{S}^{\ell}(\mathfrak{g}),$ we have \begin{align*}
    \deg \left( \Lambda(pq)\right) = k + \ell , \qquad   \Lambda(pq) -  \Lambda(p) \Lambda(q) \in \mathcal{U}_{k+ \ell-1}(\mathfrak{g}), 
\end{align*} where $\Lambda (pq) = \Lambda (p) \Lambda(q) + \text{lower order terms}$. We now pass from the filtered algebra $\mathcal{U}(\mathfrak{g})$ to its associated graded algebra \begin{align}
    \mathrm{gr} \, \mathcal{U} (\mathfrak{g}) = \bigoplus_{k \geq 0 } \mathcal{U}^k(\mathfrak{g}),
\end{align}  where $\mathcal{U}^0(\mathfrak{g}) = \mathbb{F}$. Using the Poincar\'e-Birkhoff-Witt (PBW) theorem, the linear isomorphism $\Lambda$ in \eqref{eq:symme} induces an algebra isomorphism $\tilde{\Lambda}$ between $\mathcal{S}(\mathfrak{g})$ and $\mathrm{gr} \, \mathcal{U}(\mathfrak{g})$. Let $\tilde{\Lambda}_k = \tilde{\Lambda}\vert_{\mathcal{S}^k(\mathfrak{g})} $. We then have an algebra isomorphism between $\mathcal{U}^k(\mathfrak{g})  $ and $ \mathcal{S}^k(\mathfrak{g}) $.  Define \begin{align*}
   \mathcal{C}_{\mathcal{U}^k(\mathfrak{g})}(\mathfrak{a})  &= \left\{ P \in \mathcal{U}^k(\mathfrak{g}): [X,P]_o = 0  \quad \forall X \in \mathfrak{a}\right\},\\
    \mathcal{C}_{\mathcal{S}^k(\mathfrak{g})}(\mathfrak{a})   & = \left\{p \in \mathcal{S}^k(\mathfrak{g}): \{x,p\}_o = 0  \quad \forall x \in \mathfrak{a}^*\right\}.
\end{align*}  The construction above induces the next proposition.
\begin{proposition}
\label{2.5}
    Let $\mathcal{C}_{\mathcal{U}^k(\mathfrak{g})}(\mathfrak{a}) :=  \,\mathcal{C}^k(\mathfrak{a}) $ and $ \mathcal{C}_{\mathcal{S}^k(\mathfrak{g})}(\mathfrak{a}) := \, \mathcal{C}_k(\mathfrak{a})  $. Then $ \mathcal{C}_{\mathcal{U} (\mathfrak{g})}(\mathfrak{a}) = \bigoplus_{k \geq 0} \mathcal{C}^k(\mathfrak{a})$ and $\mathcal{C}_{\mathcal{S} (\mathfrak{g})}(\mathfrak{a}) = \bigoplus_{k \geq 0}   \mathcal{C}_k(\mathfrak{a}). $ In particular, $\tilde{\Lambda}_k$ is an algebra isomorphism between $ \mathcal{C}^k(\mathfrak{a})  $ and $\mathcal{C}_k(\mathfrak{a})$.
\end{proposition}
 
Using the symmetrization map $\eqref{eq:symme}$, the commutant in the enveloping algebra $\mathcal{U}(\mathfrak{g})$ is obtained as $\mathcal{C}_{\mathcal{U}(\mathfrak{g})}(\mathfrak{a}) = \Lambda \left(\mathcal{C}_{\mathcal{S}(\mathfrak{g})}(\mathfrak{a})\right)$. Therefore, the results we obtained in the commutative setting can be therefore mapped to the non-commutative Lie algebra setting. 
 Following a procedure similar to the one leading to $\eqref{eq:bas}$, we can define the set  \begin{align}
    \hat{\textbf{Q}}_g = \bigcup_{h=1}^g \hat{\textbf{q}}_h \label{eq:14}
\end{align}  
and construct the corresponding finitely-generated algebra $ \textbf{Alg} \langle \hat{\textbf{Q}}_g \rangle .$ Here $g$ is the maximal degree of homogeneous representatives and 
$\hat{\textbf{q}}_h = \left\{P_{h,1} ,\ldots, P_{h,m_h} \right\} = \Lambda \big(\textbf{q}_h\big)$  consists of all the linearly independent and indecomposable homogeneous representatives of degree $h$, i.e., $P_{h,j} = \Lambda(p_{h,j})$ for all $1 \leq j \leq m_h$.  For any $P_h \in \hat{\textbf{q}}_h$ and $P_\ell \in \hat{\textbf{q}}_\ell,$ we define the bilinear operator $[\cdot,\cdot]: \textbf{Alg}\langle \hat{\textbf{Q}}_g\rangle \times \textbf{Alg}\langle \hat{\textbf{Q}}_g\rangle \rightarrow \textbf{Alg}\langle \hat{\textbf{Q}}_g\rangle  $ by  \begin{align} 
    [P_h ,P_\ell ]   = \sum_{h_1+ \cdots + h_j \leq \ell + h -1} \Gamma_{m_{h_1} \ldots m_{h_j}}^{h \ell} P_{h_1,m_{h_1}} \cdots P_{h_j,m_{h_j}}    , \label{eq:rec}
\end{align}  where $\Gamma_{m_{h_1} \ldots m_{h_j}}^{h \ell} \in \mathbb{F}.$ Here, again, $d$ is the $\textit{degree}$ of the quantized Poisson polynomial algebra $\eqref{eq:rec}$ given by $$d := \max_{h_1+ \cdots + h_j\leq \ell + h -1} \mathcal{N}(P_{h_1,m_{h_1}} \cdots P_{h_j,m_{h_j}} ).$$ 
Then $\left(\textbf{Alg}\langle \hat{\textbf{Q}}_g\rangle,[\cdot,\cdot]\right)$ forms a (finitely-generated) associative algebra \cite[Section 2.3]{MR4660510}, which is denoted by $\hat{\mathcal{Q}}(d).$  It is verified in \cite[Section 2.4.13]{MR1451138} that $\hat{\mathcal{Q}}(d)$ possesses the property of being finitely-generated. Notice that, for all the examples provided in the later sections of this manuscript, the associativity of $\hat{\mathcal{Q}}(d)$ can be established through the verification of the Jacobi identity.



From the construction above, we can define algebraic Hamiltonians and their corresponding superintegrable systems in $\mathcal{U}\left(\mathfrak{g}\right)$.

\begin{definition}
\label{H} \cite{MR2515551}.
  Let $\mathfrak{a} \subset \mathfrak{g}$ be a Lie subalgebra with an ordered basis $\beta_\mathfrak{a} = (X_{\ell_1},\ldots,X_{\ell_s}) $, and let $K_t$ be the Casimir invariants of $\mathfrak{g}$. Let $\hat{\textbf{Q}}_g$ be the same as defined in $\eqref{eq:14}.$ An algebraic Hamiltonian with respect to $\mathfrak{a}$ is given by \begin{align}
      \hat{\mathcal{H}} = \sum_{\ell_1,\ldots,\ell_s}^s \Gamma_{\ell_1 \ldots \ell_s}X_{\ell_1}\cdots X_{\ell_s}+   \sum_t \zeta_t K_t, \label{eq:quHa}
  \end{align} where  $\ell_1,\ldots,\ell_s$ are indices for the generators of $\mathfrak{a},$ $\Gamma_{\ell_1 \ldots \ell_s}$ and $\zeta_t$ are constants. For all algebraically independent generators in $\hat{\mathcal{Q}}_g$, $S = \{\hat{\mathcal{H}}\} \cup \{ P_k\}_{ k \geq 1}^N $ forms an $\textit{algebraic}$ $\textit{superintegrable}$ $\textit{system}$ with the underlying symmetry algebra $ \hat{\mathcal{Q}}(d) $, where $N$ is the number of algebraically independent generators.
\end{definition} 

\begin{remark}
\label{3.7}
(i) Notice that using the construction above, we can determine a unique set of linearly independent and indecomposable polynomials that generate the associative algebra $ \hat{\mathcal{Q}}(d) $. However, by the definition of centralisers, there exists a family of Hamiltonians $\hat{\mathcal{H}}$ in $\mathcal{U}(\mathfrak{a}) $ such that $\big[\hat{\mathcal{H}},\hat{\mathcal{Q}}(d)\big] = 0.$ To make $\hat{\mathcal{H}}$ a quadratic differential operator, we can choose a suitable realisation of the subalgebra $\mathfrak{a}$ such that the first summation in $\hat{\mathcal{H}}$ does not produce differential operators of order greater than 2. 

 (ii) Using the isomorphism $\Lambda$, for any  $p  \in \mathcal{C}_{\mathcal{S}(\mathfrak{g})}(\mathfrak{a})$, there exists a $P = \Lambda(p) \in \mathcal{C}_{\mathcal{U}(\mathfrak{g})}(\mathfrak{a})$ such that $[\hat{\mathcal{H}},P] = \Lambda(\{\mathcal{H},p    \}) = 0,  $ where $\mathcal{H}$ is a free Hamiltonian of the form \begin{align}
      \mathcal{H}  =\sum_{\ell_1,\ldots,\ell_s}^s \Gamma_{\ell_1 \ldots \ell_s}x_{\ell_1}\cdots x_{\ell_s}   + \sum_t \zeta_t \mathcal{C}_t,\label{eq:Hamil}
 \end{align} where  $\ell_1,\ldots,\ell_s$ are indices for the generators of $\mathfrak{a}^*,$ $\Gamma_{\ell_1 \ldots \ell_s}$ and $\zeta_t$ are constants. Here, $\mathcal{C}_t$ is the Casimir invariants of $\mathfrak{g}^*.$ Moreover, the number of functionally independent integrals of motion can be determined from $\eqref{eq:func}$. 
\end{remark}

We now have a closer look at the central elements of the Poisson enveloping algebra $\mathcal{U}\left(\mathcal{Q}(d)\right)$. See, for instance, \cite[Section 2]{MR3659329}.   A $\textit{Casimir}$ of $\mathcal{Q}(d)$ is a polynomial $K$ in $\mathcal{U}\big(\mathcal{Q}(d)\big)$ such that $\{K,p\} = 0$ for any $p \in \mathcal{U}\big(\mathcal{Q}(d)\big).$ Let $\textbf{a}_i = \ \alpha_{i,1} \cdots \alpha_{i,m_i}$ and $\textbf{q}_1^{\textbf{a}_1}\ldots\textbf{q}_g^{\textbf{a}_g} = p_{1,1}^{\alpha_{1,1}} \cdots p_{1,m_1}^{\alpha_{1,m_1}} \cdots p_{g,1}^{\alpha_{g,1}} \cdots p_{g,m_g}^{\alpha_{g,m_g}}$. The generic form of a degree $h$-polynomial $K_{\mathcal{Q}(d)}^h \in \mathcal{U}_h\big(\mathcal{Q}(d)\big)$ has the expression \begin{align}
    K_{\mathcal{Q}(d)}^h = \sum \Gamma_{1 \ldots m_g}\textbf{q}_1^{\textbf{a}_1}\cdots\textbf{q}_g^{\textbf{a}_g} , \text{ }\qquad  \sum_i^g \textbf{a}_i = \sum_i^g \sum_{j=1}^{m_i}\alpha_{i,j} = h, \label{eq:polycas}
\end{align} where $\mathcal{U}_h\left(\mathcal{Q}(d)\right) = \mathrm{span}\left\{p_1^{i_1} \cdots p_g^{i_g}: i_1+ \cdots + i_g \leq  h\right\}.$   For any $p_{h,m_h} \in \mathcal{Q}(d),$ to determine the linearly independent and indecomposable Casimir invariants, it is sufficient to determine the coefficients of $K_{\mathcal{Q}(d)}^h$ in $\eqref{eq:polycas}$ such that  \begin{align}
    \left\{K_{\mathcal{Q}(d)}^h,p_{h,m_h}\right\} = 0. \label{eq:casipde}
\end{align}  
Note that similar to the method used to determine $\textbf{Q}_g$ earlier, we must exclude the decomposable Casimir invariants. 
In this work, we identify all functionally independent Casimirs through direct calculations.

\section{Conformal algebra $\mathfrak{c}(2)$ and its conjugacy classes}
\label{3}

This section examines the algebraic properties of $\mathfrak{c}(2)$. Given that $\mathfrak{c}(2) \cong \mathfrak{so}(3,1)$ represents a Lie algebra isomorphism, it is beneficial to take advantage of the classification of the subalgebras of $\mathfrak{so}(3,1)$. The Appendix provides an overview of how subalgebras of $\mathfrak{so}(3,1)$ are classified. Here, we use the conclusions drawn from Corollary $\ref{clas}$ to identify all subalgebras of $\mathfrak{c}(2)$ based on a specified basis $\eqref{eq:real}$. We begin with the commutator relations of $\mathfrak{c}(2)$.

\begin{definition}
\label{confor}
  The $\textit{conformal algebra}$ $\mathfrak{c}(2) $ in 2-dimensional space with coordinates $(x,y)$ is generated by $6$ elements $ \{X_j\}_{j=1}^6$ which satisfy the commutation relations \cite{MR3988021},
  \begin{center}
      \begin{tabular}{|c|c c c c c c| }
      \hline
           & $X_1$ & $X_2$ & $X_3$ & $X_4$ & $X_5$ & $X_6$   \\
           \hline
    $X_1$       & $0$ & $0$ & $-X_2$ & $X_1$& $2X_4$ & $2X_3$  \\
     $X_2$       & $0$ &  $0$ & $X_1$ &$X_2$ &$-2X_3$ & $2X_4$  \\
      $X_3$        & $ X_2$ &$-X_1$ & $0$ &$0$ &$ X_6$ & $-X_5$ \\
       $X_4$        & $-X_1$ & $-X_2$& $0$& $ 0$&$X_5$ & $X_6$ \\
      $X_5$         & $-2X_4$ & $2X_3$ & $-X_6$ & $-X_5$ &0 & 0 \\
        $X_6$     & $-2X_3$ & $-2X_4$ & $X_5$ & $-X_6$ & 0&  0\\ 
        \hline
      \end{tabular}

      \quad

    $\textbf{Table 1}$ \label{1}  
  \end{center} 
\end{definition}

Notice that we have the following realisation in terms of differential operators for the generators of   $\mathfrak{c}(2)$\begin{align}
    &X_1 = \partial_x,\quad X_2 = \partial_y, \text{ } \quad X_3 = y \partial_x- x \partial_y , \text{ }\quad X_4  = x \partial_x + y \partial_y,  \nonumber\\
    & X_5 = (x^2 -y^2) \partial_x + 2xy \partial_y, \text{ } \quad X_6 = 2 xy \partial_x +(y^2 - x^2) \partial_y.\label{eq:real}
\end{align}  

We now enumerate all subalgebras, up to conjugate classes, of $\mathfrak{c}(2)$ according to the notation of Table 1. Those subalgebras will be used in the next Section \ref{4} to construct polynomial algebras from reduction chains. We begin with one-dimensional subalgebras. Let $\mathfrak{a}_j = \mathrm{span} \{X_j\}$, $1 \leq  j \leq 6$, be a one-dimensional subalgebra of $\mathfrak{c}(2)$. From the classification of the subalgebras of $\mathfrak{so}(3,1)$ in Corollary $\ref{clas}$, we can identify $\mathfrak{a}_1  $ and $\mathfrak{a}_5  $ as one-dimensional Euclidean algebra $\mathfrak{e}(1)$. We can further show that there are in different conjugation classes. That is, for some $X,Y \in \mathfrak{c}(2)$, the conjugation classes of the one-dimensional subalgebras are 
$$  \mathfrak{a}_1 \overset{{\rm Ad}(\exp(X))}{\sim} \mathfrak{a}_2,\qquad   \mathfrak{a}_5 \overset{{\rm Ad}(\exp(Y))}{\sim} \mathfrak{a}_6,$$ Here $\exp(X), \exp(Y) \in C(2)$ and $\mathrm{Ad}(\cdot)$ is the adjoint representation of $C(2)$, where $C(2)$ is the Lie group of $\mathfrak{c}(2)$. Moreover, we have $\mathfrak{a}_3 \cong \mathfrak{o}(2)$ and $\mathfrak{a}_4 \cong \mathfrak{o}(1,1).$     From Remark $\ref{clasr}$, they can all be identified with a one-dimensional Euclidean algebra $\mathfrak{e}(1).$ Finally, there also exists a one-dimensional subalgebra $\mathfrak{a}_c$ generated by $X_c = \frac{\cos c_1}{2} X_3  + \frac{\sin c_2}{2}X_4$, $0 < c_1 < \frac{\pi}{2}$ and $\frac{\pi}{2} < c_2 < \pi$,  which is the subalgebra $F_{10}$ in Corollary $\ref{clas}$. Notice that $\mathfrak{a}_c = \mathrm{span}  \{X_c\}$ is the algebra that corresponds to a rotation about the space axis with a simultaneous boost along the same axis.



We now have a look at the $2$-dimensional subalgebras, which are divided into Abelian and non-Abelian cases.   From Table 1, it is easy to observe that $ \mathfrak{a}_{(12)}$, $\mathfrak{a}_{(34)}$, and $\mathfrak{a}_{(56)}$ are Abelian, where \begin{align}
    \begin{matrix}
        \mathfrak{a}_{(12)} = \mathrm{span} \left\{X_1,X_2: \text{ } [X_1,X_2]_o = 0\right\} ; \\
        \mathfrak{a}_{(34)} = \mathrm{span} \left\{X_3,X_4: \text{ } [X_3,X_4]_o = 0\right\} ; \\
        \mathfrak{a}_{(56)} = \mathrm{span} \left\{X_5,X_6: \text{ } [X_5,X_6]_o = 0\right\} .
    \end{matrix} \label{eq:2da}
\end{align} Furthermore, using Corollary $\ref{clas}$, it follows that $\mathfrak{a}_{(12)},\mathfrak{a}_{(56)} = \mathfrak{e}(1) \oplustilde \mathfrak{e}(1)$ are maximal Abelian nilpotent subalgebra, and $\mathfrak{a}_{(34)} = \mathfrak{o}(2) \oplustilde \mathfrak{o}(1,1)$ are orthogonally decomposable maximal Abelian subalgebra\footnote{Please refer to Definition $\ref{7.2}$ in Appendix.}, where $\oplustilde$ denotes a direct sum of the Lie algebra.  On the other hand, non-Abelian $2$-dimensional subalgebras of $\mathfrak{c}(2)$ are isomorphic to those algebras with the form of $\mathfrak{a} = \mathrm{span}\{X,Y: [X,Y] = X\}.$ In particular, in the basis $\eqref{eq:real},$ we have \begin{align}
  \mathfrak{a}_{(14)} &=  \mathrm{span} \{X_1,X_4: [X_1,X_4]_o = X_1\}, \text{ } \quad\mathfrak{a}_{(24)} = \mathrm{span} \{X_2,X_4: [X_2,X_4]_o =  X_2\}, \label{eq:2dc} \\
  \nonumber
\mathfrak{a}_{(45)} & =   \mathrm{span} \{X_5,X_4: [X_4,X_5]_o =  X_5\} , \text{ } \quad\mathfrak{a}_{(46)}=\mathrm{span} \{X_6,X_4: [X_4,X_6]_o =  X_6\}.
\end{align}  Notice that all non-Abelian $\mathfrak{a}_{(ij)}$ are rank-one subalgebras, which are isomorphic to the Borel subalgebra $\mathfrak{b}_0$ of $\mathfrak{sl}(2)$.

There are $5$ non-isomorphic three-dimensional subalgebras on the basis $\beta_{\mathfrak{c}(2)}$. They are given by \begin{align}
\nonumber
  &   \mathfrak{a}_{(123)} = \mathrm{span}\{X_1,X_2,X_3\} \cong \mathfrak{a}_{(563)} = \mathrm{span} \{X_5,X_6,X_3\} \cong \mathfrak{e}(2),\\
  \nonumber
  & \mathfrak{a}_{(145)} = \mathrm{span}\{X_1,X_4,X_5\} \cong \mathfrak{a}_{(246)} = \mathrm{span} \{X_2,X_4,X_6\} \cong \mathfrak{sl}(2) ,\\
  & \mathfrak{a}_{(124)} = \mathrm{span}\{X_1,X_2,X_4\} \cong \mathfrak{a}_{(564)} = \mathrm{span} \{X_5,X_6,X_4\} \cong \mathfrak{iso}(2), \label{eq:3dc} \\
  \nonumber
  & \mathfrak{a}_{(3,6,5)}^+ = \mathrm{span} \left\{X_3,\frac{X_6 + X_2}{2}, \frac{X_1 + X_5}{2}\right\}\cong \mathfrak{su}(2),\\
  \nonumber
  & \mathfrak{a}_{(3,6,5)}^- = \mathrm{span} \left\{X_3,\frac{X_2- X_6}{2}, \frac{X_1 - X_5}{2}\right\}\cong \mathfrak{su}(1,1) \cong \mathfrak{sl}(2)
\end{align} and \begin{align*}
    \mathfrak{a}_{(1,2,4)} = \mathrm{span}\left\{\frac{\cos c_1}{2} X_3  + \frac{\sin c_2}{2}X_4,X_1,X_2\right\}, 
\end{align*} which is isomorphic to  $\mathfrak{n}_{3,3}. $ Please refer to \cite{MR3184730} for more details.

It is clear that the Borel subalgebra $\mathfrak{b}$ of $\mathfrak{c}(2)$ is a subalgebra of dimension $4$, which is constructed using the Cartan subalgebra and two positive simple roots.

\section{Construction of polynomial algebras from reduction chains of $\mathfrak{c}(2)$}
\label{4}

 In this Section \ref{4}, we construct symmetry algebras generated from the reduction chains of the subalgebras of $\mathfrak{c}(2)$. From the classification in Section $\ref{3}$, they are $\mathfrak{a}_j \subset \mathfrak{c}(2),$ $ \mathfrak{a}_c \subset \mathfrak{c}(2),$ $\mathfrak{a}_{(ij)} \subset \mathfrak{c}(2)$, $\mathfrak{a}_{(ijk)} \subset \mathfrak{c}(2)$, $\mathfrak{a}_{(3,6,5)}^\pm \subset \mathfrak{c}(2)$,  $\mathfrak{a}_{(1,2,3)} \subset \mathfrak{c}(2)$ and $\mathfrak{b} \subset \mathfrak{c}(2) $ with $1 \leq i \neq j \neq k \leq 6.$ In the following, we will calculate the symmetric algebra and its corresponding quantised algebra case by case.
 
 Let $(x_1,\ldots,x_6)$ be the coordinate functions of $\mathfrak{c}^*(2)$. For our purpose, we have to solve the systems of PDEs \begin{align}
    \{p(x_1,\ldots,x_6),x_j\}_o  = 0  \label{eq:PDE}
\end{align} under the Poisson-Lie bracket defined in $\eqref{eq:Poisson}$, where  $\displaystyle p(x_1,\ldots,x_6) = \sum_{h=1}^{\deg p} p_h(x_1,\ldots,x_6)$ with  $$\displaystyle p_h(x_1,\ldots,x_6) = \sum_{i_1+ \cdots + i_6 = h  }\Gamma_{i_1,\ldots,i_6}\,x_1^{i_1}\cdots x_6^{i_6}$$ being homogeneous monomials of degree $h$. Here $\Gamma_{i_1,\ldots,i_6}$ are constants. Then $ \{p(x_1,\ldots,x_6),x_i\}_o =0$ if and only if $ \{p_h(x_1,\ldots,x_6),$ $x_i\}_o =0$ for each $1 \leq h \leq  \deg p .$ In what follows, all the linearly independent and indecomposable monomial solutions of $\eqref{eq:PDE}$ are denoted by $\varphi_j(x_1,\ldots,x_6)$.  Note that $\mathfrak{c}^*(2)$ has two quadratic Casimir invariants \begin{align}
    \mathcal{C}_1 = x_3^2 - x_4^2 + x_1 x_5 + x_2 x_6, \quad \mathcal{C}_2 = 2 x_3 x_4 + x_2x_5 - x_1 x_6, \label{eq:Casi}
\end{align} and by Definition $\ref{H},$ the algebraic Hamiltonian is given by \begin{align}
 \mathcal{H} =\sum_{\ell_1 \ldots \ell_6}^{\dim \mathfrak{a}^*} \Gamma_{\ell_1 \ldots \ell_6} x_{\ell_1}\cdots x_{\ell_6} +     \gamma_1 \mathcal{C}_1 + \gamma_2 \mathcal{C}_2, \label{eq:1}
\end{align} where $x_{\ell_1} , \ldots  x_{\ell_6} \in \mathfrak{a}^*$, $\Gamma_{\ell_1 \ldots \ell_6}$, $\gamma_1$ and $\gamma_2$ are constants.

\subsection{Symmetry algebras from one-dimensional subalgebras} 
\label{4.1}

In the previous Section \ref{3}, we determined all possible subalgebras, up to conjugate, of $\mathfrak{c}(2)$ on the basis $\eqref{eq:real}$. 
As mentioned in Section $\ref{3},$ the one-dimensional subalgebras $\mathfrak{a}_1,\mathfrak{a}_2,\mathfrak{a}_5,\mathfrak{a}_6$ are realised in $ \mathfrak{e}(1)$, while $\mathfrak{a}_3$ and $\mathfrak{a}_4$ are realised in $\mathfrak{o}(2)$ and $ \mathfrak{o}(1,1)$, respectively. In this Subsection \ref{4.1}, we determine the Poisson polynomial algebras generated from the reductive chains $\mathfrak{a}_j \subset \mathfrak{g}$ for all $1 \leq j \leq 6$ and $\mathfrak{a}_c \subset \mathfrak{g}.$ Recall that $\mathfrak{a}_j$ is a one-dimensional subalgebra generated by $X_j$ and $\mathfrak{a}_c = \mathrm{span}\{X_c\} $ regardless of whether conjugated or not.  As demonstrated in the following subsections, all the polynomial algebras generated from one-dimensional reduction chains $\mathfrak{a}_j \subset \mathfrak{g}$ are quadratic, while the polynomial algebra generated from the reduction chain $\mathfrak{a}_c \subset \mathfrak{g}$ is Abelian.

\subsubsection{Symmetry algebra from $\mathfrak{a}_1  $}
\label{4.1.1}

We consider the subalgebra $\mathfrak{a}_1$ and construct the commutant $\mathcal{C}_{\mathcal{S}(\mathfrak{c}(2))}(\mathfrak{a}_1)$. By the construction in Section \ref{2}, we need to find the solution of the following PDE in the form of $\eqref{eq:ci}$ \begin{align}
    \left\{p(x_1,\ldots,x_6),x_1\right\}_o = x_2 \dfrac{\partial p}{\partial x_3} - x_1 \dfrac{\partial p}{\partial x_4} - 2x_4 \dfrac{\partial p}{\partial x_5} -2 x_3 \dfrac{\partial p}{\partial x_6}= 0. \label{eq:a1}
\end{align} Using Mathematica computational packages, we find that, upto degree $6$, there are $157$ linearly independent homogeneous polynomials. Not all those 157 
 solutions are indecomposable and we find that a solution in this set is decomposable if their degree is greater than or equal to 3. Since $\mathcal{C}_{\mathcal{S}(\mathfrak{c}(2))}(\mathfrak{a}_1)$ is finitely generated, we conclude that the maximal degree of linearly independent and indecomposable solutions (polynomials) is $2.$   It is shown that $\eqref{eq:a1}$ has the following indecomposable and linearly independent homogeneous polynomial solutions \begin{align}
\nonumber
    \varphi_1  & = x_1, \quad \varphi_2  = x_2, \quad \varphi_3 = x_1 x_3 + x_2 x_4, \\
    \varphi_4 &    = x_2 x_6 + x_3^2 , \quad \varphi_5 = x_1 x_6 -2 x_3 x_4 - 2x_2 x_5, \quad \varphi_6 = x_4^2 - x_1 x_5.  \label{eq:s1}
\end{align}    
 In the later sections, we will only present linearly independent and indecomposable solutions.

 
 Let $A_j = \varphi_j$ for all $1 \leq j \leq 6 $, and let $\textbf{Q}_2^{(1)} = \{A_1,\ldots,A_6\}$. We now show that the set $ \textbf{Q}_2^{(1)} $ forms a finitely-generated algebra and the elements in the set commute with a Hamiltonian.  From $\eqref{eq:1},$ the possible Hamiltonian in this case is then given by $$ \mathcal{H}_1  = \alpha x_1 + \gamma_1 \mathcal{C}_1 + \gamma_2 \mathcal{C}_2, $$ where $\mathcal{C}_1$ and $\mathcal{C}_2$ are the Casimir invariants of $\mathfrak{c}^*(2)$ defined in $\eqref{eq:Casi},$  and $\alpha = \Gamma_{10000}$ is a constant. In the following, without further state, $\gamma_1,\gamma_2$ are arbitrary coefficients for the Casimir functions. In subsequent sections, we will streamline the notation for coefficients in the induced brackets by referring to them as $\alpha, \beta$, etc.
It is easy to check $\{\mathcal{H}_1,A_j\} = 0$ for $1 \leq j \leq 6$. The linearly independent and indecomposable polynomial integrals $A_j$ form the quadratic Poisson algebra $  \textbf{Alg} \left\langle \textbf{Q}_2^{(1)} \right\rangle$, with the following non-zero brackets \begin{align}
   \{A_2,A_3\} & = A_1^2 + A_2^2,\quad \text{ } \{A_2,A_4\} =  -\{A_2,A_6\} = 2 A_3 , 
   \nonumber\\
     \{A_3,A_4\} &=- A_1 A_5 -2 A_2 A_6 = -  \{A_3,A_6\}  .   \label{eq:} 
\end{align} By a direct calculation, we can then verify that $\left\{A_i,\left\{A_j,A_k\right\}\right\} +\left\{A_j,\left\{A_k,A_i\right\}\right\}+\left\{A_k,\left\{A_i,A_j\right\}\right\}= 0$ holds for all $ 1 \leq i,j,k \leq 6$. Hence $\mathcal{Q}_1(2) = \big(\textbf{Alg} \left\langle \textbf{Q}_2^{(1)} \right\rangle,\{\cdot,\cdot\} \big)$ defines a Poisson algebra.  Notice that $\mathcal{Z}\left(\mathcal{Q}_1(2)\right) =\textbf{Alg} \langle A_1,A_5\rangle.$ Here $\mathcal{Z} \left(\cdot\right)$ is the center of Poisson algebra. Moreover, a routine computation shows that the rank of the Jacobian matrix $  J(A_1,\ldots,A_6)  $ is equal to $5$. Thus integrals in $\textbf{Q}_2^{(1)}$ are functionally dependent. Then we may choose $5$ of them such that $S_1 =  \{\mathcal{H}_1 = \alpha A_1 + \gamma_1 A_5,A_2,A_3,$ $A_4,A_6\}$ defines a superintegrable system.

By the PBW theorem, the Casimirs of the quadratic Poisson algebra $\mathcal{Q}_1(2)$ have the form \begin{align}
    K_{\mathcal{Q}_1(2)}^h =\sum_{i_1+ \cdots + i_6 = h} \Gamma_{i_1,\ldots,i_6} A_1^{i_1} \cdots A_6^{i_6} \in \mathcal{U}_h\left(\mathcal{Q}_1(2)\right), \label{eq:cas1}
\end{align}   where $\Gamma_{i_1,\ldots,i_6}$ are constants and $\mathcal{U}_h(\mathcal{Q}_1(2))$ is an enveloping Poisson quadratic algebra with polynomials upto degree $h$. 
We now find the explicit expressions for $K_{\mathcal{Q}_1(2)}^h$ such that $\{K_{\mathcal{Q}_1(2)}^h,A_j\} = 0$ for all $j.$ By a direct calculation, we can write $K_{\mathcal{Q}_1(2)}^h$ in terms of monomials as follows  \begin{align*}
K_{\mathcal{Q}_1(2)}^1 = & a_1 A_1 + a_2 \left( A_4 + A_6\right) + a_3 A_5,\\
    K_{\mathcal{Q}_1(2)}^2 = &   a_1A_1+a_2\left(A_4+A_6\right)+a_3A_5+ a_4A_1^2+ a_5A_1 (A_4+A_6)+a_6 A_1 A_5+a_7(A_4+A_6)^2 \\
    & +a_8A_5 (A_4+A_6)+ a_9A_5^2\\
    = & K_{\mathcal{Q}_1(2)}^1 + \left( K_{\mathcal{Q}_1(2)}^1 \right)^2 - \left[\frac{(a_5A_1+a_8A_5)}{a_2} \left(K_{\mathcal{Q}_1(2)}^1-(a_1A_1+a_3A_5)\right)+ a_6 A_1 A_5   \right],\\
K_{\mathcal{Q}_1(2)}^3  = & a_1A_1+a_4A_1^2+a_2\left(A_4+A_6\right)+ a_3 A_5+a_5A_1 (A_4+A_6)+a_6A_1 A_5+a_7(A_4+A_6)^2 \\
&+a_8A_5 (A_4+A_6)+a_9A_5^2+ a_{10} \left((A_1^2+A_2^2)A_6+A_2 A_5 A_1+A_3^2\right)+a_{11}A_1^3+a_{12}A_1^2 (A_4+A_6)\\
& +a_{13}A_1^2 A_5+a_{14}A_1 (A_4+A_6)^2+a_{15}A_1 A_5 (A_4+A_6)+a_{16}A_1 A_5^2+a_{17}(A_4+A_6)^3\\
& +a_{18}A_5 (A_4+A_6)^2+a_{19}A_5^2 (A_4+A_6)+a_{20}A_5^3 \\
=& K_{\mathcal{Q}_1(2)}^2 + \left( K_{\mathcal{Q}_1(2)}^1 \right)^3-\frac{1}{a_2}\left(2 a_{12} A_1^2 + 5 a_{15} A_1 A_5 + a_{19}A_5^2\right) \left(K_{\mathcal{Q}_1(2)}^1 -(a_1A_1+a_3A_5)\right) \\
& -2\left(a_{14}A_1+a_{18}A_5\right)  \left(K_{\mathcal{Q}_1(2)}^1 -(a_1A_1+a_3A_5)\right)^2 -2 A_1A_5\left(a_{13}A_1  + a_{16} A_5 \right) \\
& + a_{10} \left( (A_1^2+A_2^2)A_6+A_2 A_5 A_1+A_3^2\right) \\
K_{\mathcal{Q}_1(2)}^4  = & K_{\mathcal{Q}_1(2)}^3 + a_{21}A_1^4 a_{22}A_1 \left((A_1^2+A_2^2)A_6+A_2 A_5 A_1+A_3^2\right) + a_{23}A_1 A_5 (A_4+A_6)+a_{24}(A_4+A_6) \\
&\times\left( (A_1^2+A_2^2)A_6+A_2 A_5 A_1+A_3^2\right)+ a_{25}A_5  \left( (A_1^2+A_2^2)A_6+A_2 A_5 A_1+A_3^2\right) + a_{26}A_1^3 (A_4+A_6) \\
&+a_{27}A_1^3 A_5+a_{28}A_1^2 (A_4+A_6)^2+a_{29}A_1^2 A_5 (A_4+A_6)+a_{30}A_1^2 A_5^2+a_{31}A_1 (A_4+A_6)^3\\
&+a_{32}A_1 A_5 (A_4+A_6)^2 
  +a_{33}A_1 A_5^2 (A_4+A_6)+a_{34}A_1 A_5^3+a_{35}(A_4+A_6)^4+a_{36}A_5 (A_4+A_6)^3\\
  & +a_{37}A_5^2 (A_4+A_6)^2+a_{38}A_5^3 (A_4+A_6)+a_{39}A_5^4 \\
 =& K_{\mathcal{Q}_1(2)}^3 +\left( K_{\mathcal{Q}_1(2)}^1\right)^4 + f_4( (A_1^2+A_2^2)A_6+A_2 A_5 A_1+A_3^2, A_1,A_4 + A_6, A_5),
\end{align*} Here, for convenience, the real coefficients $\Gamma_{i_1,\ldots,i_6}$ in $\eqref{eq:cas1}$ are replaced by $a_t,$ $1 \leq t \leq 39,$ i.e., $a_1 = \Gamma_{1,0,0,0,0,0}, a_2 = \Gamma_{0,0,0,1,0,1}, $etc.  Note that associated with each of coefficients is a linearly independent monomial. It is clear that $  K_{\mathcal{Q}_1(2)}^{h+1} = K_{\mathcal{Q}_1(2)}^h +$ some functions for $h \geq 1.$ Inductively, we can show that \begin{align*}
 K_{\mathcal{Q}_1(2)}^h  & = K_{\mathcal{Q}_1(2)}^{h-1} +\left( K_{\mathcal{Q}_1(2)}^1\right)^h + f_h\left( (A_1^2+A_2^2)A_6+A_2 A_5 A_1+A_3^2, A_1,A_4 + A_6, A_5\right) 
\end{align*} for $ h \geq 3.$ Notice that $f_3\left( (A_1^2+A_2^2)A_6+A_2 A_5 A_1+A_3^2, A_1,A_4 + A_6, A_5\right) = (A_1^2+A_2^2)A_6+A_2 A_5 A_1+A_3^2.$ Then the functionally independent Casimir operators are \begin{align*}
    K_{\mathcal{Q}_1(2)}^{1,1} = A_1, \quad\text{ } K_{\mathcal{Q}_1(2)}^{1,2}= A_5, \quad\text{  }   K_{\mathcal{Q}_1(2)}^{1,3} = A_4 + A_6, \quad\text{  } K_{\mathcal{Q}_1(2)}^{3,1} = (A_1^2+A_2^2)A_6+A_2 A_5 A_1+A_3^2.
\end{align*}
Using the classical version of the realisation $ \eqref{eq:real}$, let $R(A_j) = R_j,$ these Casimir functions become \begin{align*}
    R\left( K_{\mathcal{Q}_1(2)}^1\right) &= R_1, \\
    R\left( K_{\mathcal{Q}_1(2)}^2\right) & = a_1 R_1 + a_2 R_1^2 =  a_1R\left( K_{\mathcal{Q}_1(2)}^1\right) + a_2R^2\left(K_{\mathcal{Q}_1(2)}^1\right), \\ 
    R\left( K_{\mathcal{Q}_1(2)}^3\right) & =R\left( K_{\mathcal{Q}_1(2)}^2\right) + a_3 R_1^3 + a_4 R_3^2 + a_5 (R_1^2+R_2^2)R_6, \\
     R\left( K_{\mathcal{Q}_1(2)}^4\right)& = R( K_{\mathcal{Q}_1(2)}^3) + R_1 \left(R(K_{\mathcal{Q}_1(2)}^3)- a_7R(K_{\mathcal{Q}_1(2)}^2) \right).
\end{align*}  Noticing $   K_{\mathcal{Q}_1(2)}^{h+1} = K_{\mathcal{Q}_1(2)}^h + $ some polynomials, we can show that $$R(K_{\mathcal{Q}_1(2)}^h) = R(K_{\mathcal{Q}_1(2)}^{h-1}) + \left(\sum_{l=1}^{h-3} R_1^l \right)\left(R(K_{\mathcal{Q}_1(2)}^3)- R(K_{\mathcal{Q}_1(2)}^2) \right)  $$ for any $h \geq 4.$ Then the realised functionally independent Casimir are 
  $R_1 $ and $ R\left(K_{\mathcal{Q}_1(2)}^{3,1}\right) = (R_1^2+R_2^2)R_6 +R_3^2.$

The corresponding quantum quadratic algebra is generated by the integrals \begin{align*}
   \hat{A}_1 & = X_1, \quad \hat{A}_2 = X_2, \quad \hat{A}_3 = X_1 X_3 + X_2 X_4, \quad \hat{A}_4 = X_2 X_6 + X_3^2 - X_4,  \\
   \hat{A}_5 &= X_3 + \frac{1}{2}(X_2 X_5 -X_1 X_6) +X_3X_4, \quad \text{ } \hat{A}_6  = X_4 -X_1 X_5 + X_4^2 .
\end{align*} These integrals satisfy the commutation relations \begin{align*}
   [\hat{A}_2,\hat{A}_3] & = \hat{A}_1^2 + \hat{A}_2^2,\quad \text{ } [\hat{A}_2,\hat{A}_4] = -[\hat{A}_2,\hat{A}_6]  = 2 \hat{A}_3 , \\
   [\hat{A}_3,\hat{A}_4] & = 2\left( \hat{A}_2 \hat{A}_6 - \hat{A}_3 \right) + 2 \hat{A}_1 \hat{A}_5 =  
   [\hat{A}_3,\hat{A}_6]   . 
 \end{align*} It is clear that these commutation relations form the quadratic algebra $\hat{\mathcal{Q}}_1(2)$. Analogous to the classical case, we can verify that the Jacobi identity \(\left[\hat{A}_i,\left[\hat{A}_j,\hat{A}_k\right]\right] +\left[\hat{A}_j,\left[\hat{A}_k,\hat{A}_i\right]\right]+\left[\hat{A}_k,\left[\hat{A}_i,\hat{A}_j\right]\right]= 0\) is satisfied for all $1 \leq i,j,k \leq 6$. Therefore, \(\hat{\mathcal{Q}}_1(2)\) is an associative algebra.

\subsubsection{Symmetry algebra from $\mathfrak{a}_2$}

We now consider subalgebra $\mathfrak{a}_2.$ Using the similar procedure, for $p(x_1,\ldots,x_6) \in \mathcal{S}(\mathfrak{c}(2))$ the commutant of $\mathfrak{a}_2$ satisfies the following PDE \begin{align*}
 \{p(x_1,\ldots,x_6),x_2\}_o =   -x_1 \dfrac{\partial p}{\partial x_3} - x_2 \dfrac{\partial p}{\partial x_4} + 2 x_3 \dfrac{\partial p}{\partial x_5} - 2x_4 \dfrac{\partial p}{\partial x_6} = 0.
\end{align*}  All the linearly independent and indecomposable solutions are given by  \begin{align*}
  & \varphi_1  = x_1, \quad\text{ }   \varphi_2  = x_2, \quad\text{ }  \varphi_3  = x_4 x_1 -x_2 x_3, \\
  &  \varphi_4  = x_1 x_5 + x_3^2,  \quad\text{ }  \varphi_5  =  x_1 x_6   - x_2   x_5 - 2 x_3 x_4, \quad \text{ }  \varphi_6 = x_2 x_6 - x_4^2  
\end{align*} 

Let $\varphi_j = A_j$ for all $1 \leq j \leq 6,$ and let $\textbf{Q}_2^{(2)} = \{A_1,\ldots,A_6\}$ be a finitely generating set. We now demonstrate that $\textbf{Alg} \left\langle\textbf{Q}_2^{(2)}\right\rangle$ forms a quadratic algebra and its elements commute to a Hamiltonian in $\mathcal{S}(\mathfrak{a}_1).$ From $\eqref{eq:1},$ the possible Hamiltonian for this system is then given by $$ \mathcal{H}_2 = \alpha x_2 + \gamma_1 \mathcal{C}_1 + \gamma_2 \mathcal{C}_2. $$  Here $\alpha = \Gamma_{010000}$ is a constant.
It is easy to check that $\{\mathcal{H}_2,A_j\} = 0$ for $1 \leq j \leq 6.$ The quadratic algebra $\textbf{Alg} \langle \textbf{Q}_2^{(2)} \rangle$ is generated by these integrals, possessing the following non-zero Poisson brackets  \begin{align*}
   \{A_1,A_3\} &= A_1^2 + A_2^2,\quad \text{ } - \{A_1,A_4\} = \{A_1,A_6\}  = -2 A_3 , \\
   \{A_3,A_4\} &=  A_2A_5 - 2 A_1 A_6 = - \{A_3,A_6\}.   
\end{align*}  Moreover, direct calculations show that the Jacobi identity for these polynomial vanish. Hence $ \mathcal{Q}_2(2) = \big(\textbf{Alg} \left\langle\textbf{Q}_2^{(2)}\right\rangle, \{\cdot,\cdot\} \big) $ is a Poisson quadratic algebra. By checking the Jacobian matrix for all the polynomial in $\textbf{Q}_2^{(2)}$, we find that the number of functionally independent elements is $5$.  Therefore, the corresponding superintegrable system is $S_2 = \{\mathcal{H}_2 =\alpha A_2 + \gamma A_5, A_1,A_3,A_4,A_6\}.$   Similarly to the arguments in Subsection $\ref{4.1.1}$, the functionally independent Casimir operators are  \begin{align*}
      K_{\mathcal{Q}_2(2)}^{1,1} = A_2, \quad\text{ } K_{\mathcal{Q}_2(2)}^{1,2}= A_5,\quad \text{  }   K_{\mathcal{Q}_2(2)}^{1,3} = A_4 + A_6, \quad\text{  } K_{\mathcal{Q}_2(2)}^{3,1} = (A_1^2+A_2^2)A_6+A_2 A_5 A_1+A_3^2.  
\end{align*} Moreover, using the classical version of the realisation $ \eqref{eq:real}$, the realised functionally independent Casimir are 
  $R\left( K_{\mathcal{Q}_2(2)}^{1,1}  \right) = R_2 $ and $ R\left(K_{\mathcal{Q}_1(2)}^{3,1}\right) = (R_1^2+R_2^2)R_6 +R_3^2.$ 

The corresponding quantum algebra has the following generators \begin{align*}
    \hat{A}_1 &= X_1, \quad \hat{A}_2 = X_2 , \quad \hat{A}_3 = X_2 X_3 - X_1 X_4, \\
    \hat{A}_4 & = X_1 X_5 + X_3^2 - X_4, \text{ } \hat{A}_5 = X_3 + \frac{X_2 X_5- X_1 X_6 }{2} + X_3 X_4. \\
    \hat{A}_6 & =    2X_6X_2 - 2X_4^2 + 2 X_4.
\end{align*} They form the closed quadratic algebra $\hat{\mathcal{Q}}_2(2)$ with non-zero commutation relations \begin{align*}
    [\hat{A}_1,\hat{A}_3] &= -\hat{A}_1^2 - \hat{A}_2^2, \quad \text{ } -[\hat{A}_1,\hat{A}_4] =
   [\hat{A}_1,\hat{A}_6]    = -2\hat{A}_3,\\ 
   [\hat{A}_3,\hat{A}_4] & = -2 \left(\hat{A}_3 + \hat{A}_1\hat{A}_6 \right) + 2 \hat{A}_2 \hat{A}_5 = [\hat{A}_3,\hat{A}_6]. 
 \end{align*}  Similar to the Subsection \ref{4.1.1}, the associativity of $\hat{\mathcal{Q}}_2(2)$ has been checked by verifying the Jacobi identity.

\subsubsection{Symmetry algebra from $\mathfrak{a}_3$}

The commutant of the subalgebra $\mathfrak{a}_3$ satisfies the linear PDE \begin{align}
   \left\{p(x_1,\ldots,x_6),x_3\right\}_o = - x_2 \dfrac{\partial p}{\partial x_1} + x_1 \dfrac{\partial p}{\partial x_2} - x_6 \dfrac{\partial p}{\partial x_5} + x_5 \dfrac{\partial p}{\partial x_6} = 0. \label{eq:a3}
\end{align} This PDE $\eqref{eq:a3}$ has $6$  linearly independent and indecomposable solutions as follows  \begin{align*}
    \varphi_1  &= x_3, \quad\text{ }  \varphi_2  = x_4, \quad\text{ }  \varphi_3 = x_1^2 + x_2^2, \quad\text{ }  \varphi_4   = x_1x_5 + x_2 x_6, \\ \varphi_5  & = x_1 x_6 - x_2 x_5 , \hskip 0.5cm \quad \varphi_6 = x_5^2 +  x_6^2.
\end{align*} The possible Hamiltonian of the system, from $\eqref{eq:1}$, is given by \begin{align*}
    \mathcal{H}_3   = \alpha x_3 + \gamma_1 \mathcal{C}_1 + \gamma_2 \mathcal{C}_2,
\end{align*}  where $\alpha = \Gamma_{001000}$.

Let $A_j = \varphi_j(x_1,\ldots,x_6)$ for any $ 1 \leq j \leq 6$, and let $\textbf{Q}_2^{(3)} = \{A_1,\ldots,A_6\}.$  Note that $A_4 - 2 A_1 A_2 =\mathcal{C}_2 $ is a Casimir function of $\mathfrak{c}(2)$. Hence, we may take $\mathcal{H}_3 = \alpha A_1 + \gamma \mathcal{C}_2$ as the Hamiltonian. We can further find a polynomial relation among $A_3,A_4,A_5$ and $A_6,$ which is given by $A_3A_6 = A_4^2 + A_5^2.$ Therefore, the superintegrable system is given by $S_3 = \{\mathcal{H}_3, A_2,A_3, A_5, A_6\}.$ We are now looking for its hidden symmetry algebra.  Polynomials $A_j$ form the quadratic algebra $  \textbf{Alg} \langle \textbf{Q}_2^{(3)} \rangle$ with the following non-zero bilinear operation \begin{align*}
      \{A_2,A_3\} &= -2 A_3, \quad \text{ }  \{A_2,A_6\} = 2 A_6,\\
    \{A_3,A_4\} &= 4 A_3 A_2, \quad \text{ } \{A_3,A_5\} = 4 A_3 A_1 ,\\
    \{A_3,A_6\} & = 8 \left( A_2 A_4 + A_1 A_5 \right), \\
 \{A_4,A_6\}& = 4 A_2 A_6, \quad   \{A_5,A_6\}   = 4 A_1 A_6.
\end{align*} By checking the Jacobi identity, we deduce that \begin{align*}
    \left\{A_i,\left\{A_j,A_k\right\}\right\} +\left\{A_j,\left\{A_k,A_i\right\}\right\}+\left\{A_k,\left\{A_i,A_j\right\}\right\}= 0
\end{align*} for all $ 1\leq i,j,k \leq 6.$ Hence $\mathcal{Q}_3(2) = \left(\textbf{Alg} \langle \textbf{Q}_2^{(3)} \rangle,\{\cdot,\cdot\}\right)$ is a Poisson quadratic algebra. We now determine the Casimir invariants of $\mathcal{Q}_3(2)$. By direct calculations, we find  \begin{align*}
  K_{\mathcal{Q}_3(2)}^1  =& c_1A_1, \\
  K_{\mathcal{Q}_3(2)}^2 
  = &  K_{\mathcal{Q}_3(2)}^1 + \left( K_{\mathcal{Q}_3(2)}^1\right)^2 + c_2\left(A_4-A_2^2\right)+c_3\left(A_5-2 A_1 A_2\right) -c_5\left(A_4^2+A_5^2-A_3 A_6\right), \\
  K_{\mathcal{Q}_3(2)}^3  
   = &  K_{\mathcal{Q}_3(2)}^2 - c_5 \left(A_4^2+A_5^2-A_3 A_6\right)\\
    &+ f_3\left(A_1,4 A_4 A_1^2-4 A_2 A_5 A_1-A_4^2+A_3 A_6,4 A_4 A_1^2-4 A_2 A_5 A_1+A_5^2\right), \\
  K_{\mathcal{Q}_3(2)}^4 =
 & K_{\mathcal{Q}_3(2)}^3 +c_6A_2^4 +4 c_7 A_1^2 A_2^2 + c_8\left(A_2^2-A_4\right)^2+\left(A_2^2-A_4\right) (2 A_1 A_2-A_5) \\
  & + c_9A_1 \left(4 A_4 A_1^2-4 A_2 A_5 A_1-A_4^2+A_3 A_6\right) +c_{10}A_1 \left(4 A_4 A_1^2-4 A_2 A_5 A_1+A_5^2\right)\\
  & +c_{11}\left(A_2^2-A_4\right) \left(A_4^2+A_5^2-A_3 A_6\right)  +c_{12}(2 A_1 A_2-A_5) \left(A_4^2+A_5^2-A_3 A_6\right)\\
  & +c_{13}A_1^4-c_{14}A_1^2 \left(A_4^2+A_5^2-A_3 A_6\right)+c_{15}\left(A_4^2+A_5^2-A_3 A_6\right)^2,
\end{align*} where $c_t$ are real coefficients that separate each linearly independent terms with all $1 \leq t \leq 15.$ Inductively, it is sufficient to show that for  $k\geq 4$ there exist polynomials $h_k$ of the form \begin{align*}
    K_{\mathcal{Q}_3(2)}^k = &h_k\left(A_1,  A_4-A_2^2, ~ A_5-2 A_1 A_2,~ A_4^2+A_5^2-A_3 A_6 ,\right.\\
    & \left. 4 A_4 A_1^2-4 A_2 A_5 A_1-A_4^2+A_3 A_6,~ 4 A_4 A_1^2-4 A_2 A_5 A_1+A_5^2\right).
\end{align*} Denote\begin{align*}
    K_{\mathcal{Q}_3(2)}^{1,1} = & \,   A_1  , \quad \text{  }   K_{\mathcal{Q}_3(2)}^{2,1} =  A_4-A_2^2, \quad\text{  } K_{\mathcal{Q}_3(2)}^{2,2} = A_5-2 A_1 A_2, \quad\text{ } K_{\mathcal{Q}_3(2)}^{2,3} = A_4^2+A_5^2-A_3 A_6, \\
    K_{\mathcal{Q}_3(2)}^{2,4} = & \, 4 A_4 A_1^2-4 A_2 A_5 A_1-A_4^2+A_3 A_6, \text{  }\quad K_{\mathcal{Q}_3(2)}^{2,5} =4 A_4 A_1^2-4 A_2 A_5 A_1+A_5^2
\end{align*}  and notice that we can rewrite \begin{align*}
     K_{\mathcal{Q}_3(2)}^{2,4} = 4 \left( K_{\mathcal{Q}_3(2)}^{1,1}\right)^2K_{\mathcal{Q}_3(2)}^{2,1}- K_{\mathcal{Q}_3(2)}^{2,3} + \left(K_{\mathcal{Q}_3(2)}^{2,2}\right)^2, \quad\text{ } K_{\mathcal{Q}_3(2)}^{2,5} = \left(K_{\mathcal{Q}_1(2)}^{2,2}\right)^2 - 4   \left( K_{\mathcal{Q}_3(2)}^{1,1}\right)^2K_{\mathcal{Q}_3(2)}^{2,1}.
\end{align*} Then the functionally independent Casimir functions for the quadratic algebra $\mathcal{Q}_3(2)$ are $K_{\mathcal{Q}_3(2)}^{1,1}, K_{\mathcal{Q}_3(2)}^{2,1}$, $ K_{\mathcal{Q}_3(2)}^{2,2}$ and $ K_{\mathcal{Q}_3(2)}^{2,3}$.

\vskip 0.2cm

The realisation of the polynomials $A_j$ satisfies the constraints \begin{align}
    R_4- R_2^2 = -R_3^2  ,  \text{ }\quad  R_5 = 2 R_1 R_2 ,\text{ } \quad R_3 R_6 = R_4^2 + R_5^2. \label{eq:realiz3}
\end{align} Using these constraints, the above Casimir functions in the realisation become \begin{align*}
    R\left( K_{\mathcal{Q}_3(2)}^1\right)  =& \, R_1, \quad \text{ } R \left(K_{\mathcal{Q}_3(2)}^2\right) = c_1 R_1 + c_2 R_1^2 - c_3 R_3^2, \\
    R \left(K_{\mathcal{Q}_3(2)}^3\right) = & \, R \left(K_{\mathcal{Q}_3(2)}^2\right) + c_4 R_1^3  - c_5 R_1 R_3^2 + 8 c_6 R_1^2 R_4 - 2 c_7 R_5^2 \\
    R \left(K_{\mathcal{Q}_3(2)}^4\right) 
      = & \,  R \left(K_{\mathcal{Q}_3(2)}^3\right)  -c_8 R_1^3 + 2 c_9 \left( R_1 + R_1^2 -   R \left(K_{\mathcal{Q}_3(2)}^2\right)\right)^2  .
\end{align*} It is clear that both $ R \left(K_{\mathcal{Q}_3(2)}^2\right)$ and $ R \left(K_{\mathcal{Q}_3(2)}^3\right)$ can not be expressed in terms of other $R \left(K_{\mathcal{Q}_3(2)}^h\right)$. we can show, inductively, that \begin{align*}
    R \left(K_{\mathcal{Q}_3(2)}^k\right) = \sum_{l=1}^{k-1} R_1^l + R_1^{k-3} R \left(K_{\mathcal{Q}_3(2)}^{k-1}\right) + f\left(K_{\mathcal{Q}_3(2)}^1,\ldots,K_{\mathcal{Q}_3(2)}^{k-2}\right)
\end{align*} for $k \geq 4.$  In particular, in the realisation the functionally independent Casimirs become \begin{align*}
     R \left(K_{\mathcal{Q}_3(2)}^{1,1}\right) = R_1, \text{ } \quad  R \left(K_{\mathcal{Q}_3(2)}^{2,1}\right) = -R_1^2, \quad\text{ } R \left(K_{\mathcal{Q}_3(2)}^{2,3}\right) =  \left(K_{\mathcal{Q}_3(2)}^{2,4}\right) = 0.
\end{align*}

We remark that the integrals of the corresponding quantum system are given by \begin{align*}
    \hat{A}_1 & \,= X_3, \quad\text{ }  \hat{A}_2 = X_4, \quad \text{ } \hat{A}_3 = X_1^2 + X_2^2,\quad \text{ }  \hat{A}_4 = X_2 X_5  - X_1 X_6, \\
    \hat{A}_5& \,  = X_5^2 + X_6^2,  \hskip 1.1cm \quad\text{  }\hat{A}_6 = X_5X_1 +X_6 X_2.
\end{align*}  They form the quadratic algebra $\hat{\mathcal{Q}}_3(2)$ with the following commutator relations \begin{align*}
    [\hat{A}_2,\hat{A}_3]& \, = -2 \hat{A}_3,\quad\text{ } [\hat{A}_2,\hat{A}_5] = 2 \hat{A}_5,  \\
    [\hat{A}_3,\hat{A}_4] & \, = - 4 \hat{A}_1 \hat{A}_3, \quad\text{ } [\hat{A}_3,\hat{A}_5] = -16(\hat{A}_1^2 + \hat{A}_2^2) - 8  \hat{A}_1\hat{A}_4 + 8 \hat{A}_2 \hat{A}_6,\\
    [\hat{A}_3,\hat{A}_6] & \,= 4 \hat{A}_2\hat{A}_3 + 8 \hat{A}_3,\quad 
      [\hat{A}_4,\hat{A}_5]  = -4  \hat{A}_1 \hat{A}_5, \quad
    [\hat{A}_5,\hat{A}_6]  = -4  \hat{A}_2\hat{A}_5. 
\end{align*}  By  verifying the Jacobi identity, we see that $\hat{\mathcal{Q}}_3(2)$ is associative. 
 

\subsubsection{Symmetry algebra from $\mathfrak{a}_4$}

The commutant of $\mathfrak{a}_4$ induces the following PDE \begin{align*}
    \{p(x_1,\ldots,x_6),x_4\}_o = x_1 \dfrac{\partial p}{\partial x_1} + x_2 \dfrac{\partial p}{\partial x_2} - x_5 \dfrac{\partial p}{\partial x_5}  - x_6 \dfrac{\partial p}{\partial x_6} = 0.
\end{align*} This PDE has the following linearly independent and indecomposable solutions \begin{align*}
    \varphi_1   =& \, x_3, \hskip 0.4cm \text{ }\quad  \varphi_2  = x_4, \hskip 0.3cm \text{ }\quad   \varphi_3  = x_1x_5,\\
    \varphi_4   =& \, x_1x_6,\text{ }\quad  \varphi_5  = x_2x_6, \text{ }\quad  \varphi_6 = x_2x_5
\end{align*}
The possible algebraic Hamiltonian is given by \begin{align*}
    \mathcal{H}_4  = \alpha x_4 + \gamma_1 \mathcal{C}_1 + \gamma_2 \mathcal{C}_2.
\end{align*} Here $\alpha = \Gamma_{000100} \in \mathbb{R}$. Let $\varphi_j(x_1,\ldots,x_6) = A_j$ for all $ 1 \leq j \leq 6$, and let $\textbf{Q}_2^{(4)} = \{A_1,\ldots,A_6\}.$ It is easy to verify that there is a polynomial relations among the integrals $A_j$ for all $j$ such that $ A_3 A_5 - A_4A_6 =0.$ By definition, one of the possible Hamiltonian is given by $\mathcal{H}_4 =  \alpha A_2 + \gamma_1 \mathcal{C}_1$, where $\mathcal{C}_1 = A_1^2 + A_2^2 -A_3 +A_5.$ Then $S_4 = \{\mathcal{H}_4, A_1,A_3, A_4 ,A_6\}$ defines a superintegrable system.  The integrals form the closed quadratic algebra $ \textbf{Alg} \langle \textbf{Q}_2^{(4)} \rangle$ with $\{A_2,A_j\} = 0$ for all $1 \leq j \leq 6$ and \begin{align*}
 \{A_1,A_3\} & \,= A_4 + A_6 = -\{A_1,A_5\}, \text{ }  \{A_1,A_4\}   =     \{A_1,A_6\} =A_5 -A_3, \\
 \{A_3,A_4\} & \, = 2 (A_3 A_1-  A_2 A_4) ,\text{ }  \{A_3,A_5\} =  2 (A_4  + A_6  )A_1,\text{ }  \{A_3,A_6\} = 2(A_3A_1 + A_6 A_2) \\
 \{A_4,A_5\} & \, = 2(A_5 A_1 - A_4A_2) , \text{ } \{A_4,A_6\} =  2(A_2 A_5 - A_3 A_1) ,\text{ }  \{A_5,A_6\} = -2(A_1 A_5 + A_2 A_6).
\end{align*} For all $ 1\leq i,j,k \leq 6,$   \begin{align*}
    \left\{A_i,\left\{A_j,A_k\right\}\right\} +\left\{A_j,\left\{A_k,A_i\right\}\right\}+\left\{A_k,\left\{A_i,A_j\right\}\right\}= 0.
\end{align*}   Hence $\mathcal{Q}_4(2) =\left(\textbf{Alg} \langle \textbf{Q}_2^{(4)} \rangle,\{\cdot,\cdot\}\right)$ is a Poisson quadratic algebra. This quadratic algebra admits the following central elements  \begin{align*}
 K_{\mathcal{Q}_3(2)}^1  =& \, d_1 A_2, \\
    K_{\mathcal{Q}_4(2)}^2 = & \, d_1A_2+ d_2\left(A_1^2+A_3+A_5\right) + d_3\left(-2 A_1 A_2+A_4-A_6\right) +d_4A_2^2+ d_5\left(A_3 A_5-A_4 A_6\right) \\
      K_{\mathcal{Q}_4(2)}^3 = & \,   K_{\mathcal{Q}_4(2)}^2 + d_6 A_2 \left(A_1^2+A_3+A_5\right)  + d_7A_2 (-2 A_1 A_2+A_4-A_6) + d_8\left(-4 (A_3+A_5) A_2^2 \right.\\
      & \, \left. +4 A_1 (A_6-A_4) A_2+(A_4-A_6)^2\right) +d_9A_2^3+ d_{10} A_2 (A_3 A_5-A_4 A_6) \\
K_{\mathcal{Q}_4(2)}^4 = & \,  K_{\mathcal{Q}_4(2)}^3 + d_{11}\left(A_1^4+2 (A_3+A_5) A_1^2+A_3^2+A_5^2+2 A_4 A_6\right)  - d_{12}\left(A_1^2+A_3+A_5\right) (2 A_1 A_2-A_4+A_6)\\
&+d_{13}(2 A_1 A_2-A_4+A_6)^2  +d_{14}A_2^2 \left(A_1^2+A_3+A_5\right)+d_{15}A_2^2 (-2 A_1 A_2+A_4-A_6) \\
& + d_{16} A_2 \left(-4 (A_3+A_5) A_2^2 +4 A_1 (A_6-A_4) A_2+(A_4-A_6)^2\right)+ d_{17}\left(A_1^2+A_3+A_5\right) (A_3 A_5-A_4 A_6) \\
& +d_{18}(2 A_1 A_2-A_4+A_6) (A_4 A_6-A_3 A_5)+d_{19}A_2^4+d_{20}A_2^2 (A_3 A_5-A_4 A_6) ,
\end{align*} where $d_t$ are real coefficients with $1 \leq t \leq 20$. Inductively, we can prove that there exist polynomials $g_k$ generated by monomials as follows \begin{align*}
   K_{\mathcal{Q}_4(2)}^k = & g_k\left(A_2, A_1^2+A_3+A_5,-2 A_1 A_2+A_4-A_6,A_3A_5 - A_4 A_6,\right. \\
   &-4 (A_3+A_5) A_2^2  +4  (A_6-A_4)A_1 A_2+(A_4-A_6)^2, A_1^4+2 (A_3+A_5) A_1^2+A_3^2+A_5^2+2 A_4 A_6, \\
   & 4 A_2 (A_6-A_4) A_1^3+\left((A_4-A_6)^2-4 A_2^2 (A_3+A_5)\right) A_1^2-4 A_2 (A_3+A_5) (A_4-A_6) A_1   \\
   & \left.+(A_3+A_5) (A_4-A_6)^2-4 A_2^2 \left(A_3^2+A_5^2+2 A_4 A_6\right)\right) 
\end{align*} for any $k \geq 4.$ We denote the generators of the polynomials $g_k$ by \begin{align*}
   K_{\mathcal{Q}_4(2)}^{1,1} =& A_2, \\
   K_{\mathcal{Q}_4(2)}^{2,1}= & A_1^2+A_3+A_5, \text{ } K_{\mathcal{Q}_4(2)}^{2,2} =-2A_1 A_2+A_4-A_6, \text{ }  K_{\mathcal{Q}_4(2)}^{2,3} =A_3A_5 - A_4 A_6 \\
   K_{\mathcal{Q}_4(2)}^{3,1} =& -4 (A_3+A_5) A_2^2  +4  (A_6-A_4)A_1 A_2+(A_4-A_6)^2, \\
   K_{\mathcal{Q}_4(2)}^{4,1} = & \, A_1^4+2 (A_3+A_5) A_1^2+A_3^2+A_5^2+2 A_4 A_6, \\
   K_{\mathcal{Q}_4(2)}^{5,1} = & \, 4 A_2 (A_6-A_4) A_1^3+\left((A_4-A_6)^2-4 A_2^2 (A_3+A_5)\right) A_1^2-4 A_2 (A_3+A_5) (A_4-A_6) A_1   \\
   &  +(A_3+A_5) (A_4-A_6)^2-4 A_2^2 \left(A_3^2+A_5^2+2 A_4 A_6\right).
\end{align*} Notice that there are constraints among $ K_{\mathcal{Q}_4(2)}^{5,1}$, $ K_{\mathcal{Q}_4(2)}^{4,1}$ ,$K_{\mathcal{Q}_4(2)}^{3,1}$, $K_{\mathcal{Q}_4(2)}^{2,3}$ and $ K_{\mathcal{Q}_4(2)}^{2,2}$. That is, \begin{align*}
K_{\mathcal{Q}_4(2)}^{4,1}& = \left(K_{\mathcal{Q}_4(2)}^{2,1}\right)^2 - 2  K_{\mathcal{Q}_4(2)}^{2,3}, \text{  }  K_{\mathcal{Q}_4(2)}^{3,1} = \left(K_{\mathcal{Q}_4(2)}^{2,2}\right)^2 - 4 \left(K_{\mathcal{Q}_4(2)}^{1,1}\right)^2K_{\mathcal{Q}_4(2)}^{2,1}\\
&K_{\mathcal{Q}_4(2)}^{5,1}K_{\mathcal{Q}_4(2)}^{2,1} - K_{\mathcal{Q}_4(2)}^{3,1}K_{\mathcal{Q}_4(2)}^{4,1} = 2 \left(K_{\mathcal{Q}_4(2)}^{2,2}\right)^2 \left(K_{\mathcal{Q}_4(2)}^{2,3}\right).
\end{align*} Hence the functionally independent Casimir functions are $K_{\mathcal{Q}_1(2)}^{1,1}, K_{\mathcal{Q}_4(2)}^{2,1}$, $ K_{\mathcal{Q}_4(2)}^{2,2}$ and $ K_{\mathcal{Q}_4(2)}^{2,3}$. The realisation admits the following constraints \begin{align*}
R_1^2 +R_3 + R_5 = R_2^2,\text{ } \quad R_4 - R_6 = 2 R_1 R_2  ,\text{ } \quad   R_3 R_5 = R_4 R_6. 
\end{align*} Then the realized Casimir functions are \begin{align*}
    R\left(K_{\mathcal{Q}_1(2)}^{1,1}\right)  = R_2, \text{ }\quad R\left(K_{\mathcal{Q}_1(2)}^{2,1}\right)  = R_2^2, \text{ } \quad   R\left(K_{\mathcal{Q}_1(2)}^{2,2}\right)  = R\left(K_{\mathcal{Q}_1(2)}^{2,3}\right)  = 0.
\end{align*}

The corresponding quantum version has the integrals \begin{align*}
    \hat{A}_1 & = X_3,\hskip 0.47cm \quad \hat{A}_2 = X_4, \quad \hskip 0.5cm  \hat{A}_3 = X_1 X_5 , \\
    \hat{A}_4 & = X_1 X_6, \quad \hat{A}_5 = X_2 X_6, \quad  \hat{A}_6 = X_2 X_5.
\end{align*} In addition to the zero commutators $[\hat{A}_2,\hat{A}_j] = 0$, they satisfy the quantum quadratic algebra  \begin{align*}
    [\hat{A}_1,\hat{A}_3]&= \hat{A}_4 + \hat{A}_6= -  [\hat{A}_1,\hat{A}_5], \text{ } [ \hat{A}_1,\hat{A}_4] = \hat{A}_1 \hat{A}_4+\hat{A}_2 \hat{A}_3+\frac{\hat{A}_4 \hat{A}_6-\hat{A}_3 \hat{A}_5}{2}  = [\hat{A}_1,\hat{A}_6],  \\
    [\hat{A}_3,\hat{A}_4] & =2 \hat{A}_1 \hat{A}_3-2 \hat{A}_2 \hat{A}_4-2 \hat{A}_4-2 \hat{A}_6 ,\\
    [\hat{A}_3,\hat{A}_5] & = 2 \hat{A}_1 \hat{A}_6-2 \hat{A}_2 \hat{A}_3+\hat{A}_3 \hat{A}_5-\hat{A}_4 \hat{A}_6, \text{ } [\hat{A}_3,\hat{A}_6] = 2 (\hat{A}_3 \hat{A}_1 + \hat{A}_6\hat{A}_2),  \\
    [\hat{A}_4,\hat{A}_5] & = 2(\hat{A}_5 \hat{A}_1 - \hat{A}_4 \hat{A}_2)  ,\text{ }  [\hat{A}_4,\hat{A}_6] = 2 \hat{A}_1 \hat{A}_4+2 \hat{A}_2 \hat{A}_5-\hat{A}_3 \hat{A}_5+\hat{A}_4 \hat{A}_6,\\
    [\hat{A}_5,\hat{A}_6] & =  2(\hat{A}_5\hat{A}_1 + \hat{A}_6 - \hat{A}_6\hat{A}_2 - \hat{A}_4)-2(\hat{A}_4 + \hat{A}_6) .
\end{align*}  Direct calculations show that the Jacobi identity for the commutator relations holds. Hence $\hat{\mathcal{Q}}_4(2)$ is an associative algebra.

\subsubsection{Symmetry algebra from $\mathfrak{a}_5$}
\label{4.1.5}
The commutant of $\mathfrak{a}_5$ induces the following PDE \begin{align*}
    \{p(x_1,\ldots,x_6),x_5\}_o = 2x_4 \dfrac{\partial p}{\partial x_1}-2 x_3 \dfrac{\partial p}{\partial x_2} + x_6 \dfrac{\partial p}{\partial x_3} + x_5 \dfrac{\partial p}{\partial x_4} = 0.
\end{align*} It has $6$ linearly independent and indecomposable homogeneous polynomial solutions given by \begin{align*}
    \varphi_1  & = x_5,\quad \varphi_2  = x_6 , \quad \varphi_3  = x_4^2-x_5x_1, \\
    \varphi_4  & = x_1 x_6 - x_2 x_5 -2 x_3 x_4, \text{ }  \varphi_5  = x_3^2 + x_2 x_6, \text{ }  \varphi_6 =  x_3x_5 - x_4 x_6.
\end{align*}
Similar to the discussion above, the possible choice of an algebraic Hamiltonian is  \begin{align*}
    \mathcal{H} = \alpha x_5 + \gamma_1 \mathcal{C}_1 + \gamma_2 \mathcal{C}_2.
\end{align*} Here $\alpha = \Gamma_{000010} \in \mathbb{R}$. Let $A_j = \varphi_j(x_1,\ldots,x_6)$ for all $ 1 \leq j \leq 6$, and let $\textbf{Q}_2^{(5)} = \{A_1,\ldots,A_6\}.$  They form the quadratic algebra $\textbf{Alg} \left\langle \textbf{Q}_2^{(5)} \right\rangle$ with non-zero  brackets \begin{align*}
  \{A_2,A_3\} &=-2 A_6 =-  \{A_2,A_5\} , \text{ }
    \{A_2,A_6\}  =A_1^2 + A_2^2, \\
\{A_3,A_6\} &  = A_1A_4 - 2 A_2 A_3 = - \{A_5,A_6\} .  
\end{align*} Similar to previous calculations, by checking $   \left\{A_i,\left\{A_j,A_k\right\}\right\} +\left\{A_j,\left\{A_k,A_i\right\}\right\}+\left\{A_k,\left\{A_i,A_j\right\}\right\}= 0$ via a direct calculation for all $1 \leq i,j,k \leq 6$, we conclude that $\mathcal{Q}_5(2) = \left(\textbf{Alg} \left\langle \textbf{Q}_2^{(5)} \right\rangle,\{\cdot,\cdot\}\right)$ is a Poisson algebra. Since $A_4 = \mathcal{C}_1$ and $A_2, A_3,A_6$ are functionally independent,
 the algebraic superintegrable system is given by $S_5 = \{\mathcal{H}_5 = \alpha A_1 + \gamma_1 A_4, A_2,A_3,A_5\}.$ We now provide the Casimir operators $K_{\mathcal{Q}_5(2)}$ for $\mathcal{Q}_5(2)$ by finding all the polynomial solutions in $\{K_{\mathcal{Q}_5(2)},A_j\} = 0$ for all $j$. A long and routine calculation shows that the Casimir operators of this algebra are given by \begin{align*}
K_{\mathcal{Q}_5(2)}^1 = & e_1A_1 + e_2\left(A_3 + A_5\right) + e_3 A_4 \\
K_{\mathcal{Q}_5(2)}^2 = &     K_{\mathcal{Q}_5(2)}^1+e_4A_1^2+e_5(A_3+A_5)^2+e_6A_4^2+e_7A_1 (A_3+A_5)+ e_8A_1 A_4+e_9 A_4 (A_3+A_5) \\
= & K_{\mathcal{Q}_5(2)}^1 + \left( K_{\mathcal{Q}_5(2)}^1 \right)^2 - \left[\frac{(e_7A_1+e_9A_4) }{e_2}\left(K_{\mathcal{Q}_5(2)}^1-(e_1 A_1+ e_3A_4)\right)+e_8A_1 A_4   \right] \\
K_{\mathcal{Q}_5(2)}^3 
=& K_{\mathcal{Q}_5(2)}^2 + \left( K_{\mathcal{Q}_5(2)}^1 \right)^3-\frac{1}{e_2}\left(2 e_{12} A_1^2 + 5 e_{15} A_1 A_4 + e_{19}A_4^2\right) \left(K_{\mathcal{Q}_5(2)}^1 -(e_1A_1+e_3A_4)\right) \\
& -2\left(e_{14}A_1+e_{18}A_4\right)  \left(K_{\mathcal{Q}_5(2)}^1 -(e_1A_1+e_3A_4)\right)^2 -2 A_1A_4\left(e_{13}A_1  + e_{16} A_4 \right) \\
& + e_{10} \left( A_2^2 A_3-A_5 A_1^2-A_2 A_4 A_1+A_6^2\right)  \\
K_{\mathcal{Q}_5(2)}^4   
 = & K_{\mathcal{Q}_5(2)}^3 + \left( K_{\mathcal{Q}_5(2)}^1 \right)^4 + \ell_4 \left(A_2^2 A_3-A_5 A_1^2-A_2 A_4 A_1+A_6^2, A_1,A_3+A_5,A_4\right).
\end{align*} Here $e_l$ are real coefficients for all $1 \leq l \leq 10$. Inductively, we observe that, for any $m \geq 4,$ there exist polynomials $\ell_k $ such that \begin{align*}
     K_{\mathcal{Q}_5(2)}^m =  K_{\mathcal{Q}_5(2)}^{m-1} + \left( K_{\mathcal{Q}_5(2)}^1\right)^m + \ell_m \left(A_2^2 A_3-A_5 A_1^2-A_2 A_4 A_1+A_6^2, A_1,A_3+A_5,A_4\right).
\end{align*} Hence, there are $4$ linearly independent Casimir operators \begin{align*}
     K_{\mathcal{Q}_5(2)}^{1,1} = A_1  , \text{  }  K_{\mathcal{Q}_5(2)}^{1,2} =   A_4, \text{ } K_{\mathcal{Q}_5(2)}^{1,3} = A_5 + A_3, \text{  } K_{\mathcal{Q}_5(2)}^{3,1} = A_2^2 A_3-A_5 A_1^2-A_2 A_4 A_1+A_6^2.
\end{align*} Moreover, the classical version of the realisation $ \eqref{eq:real}$ gives functionally independent Casimir $ R\left(  K_{\mathcal{Q}_5(2)}^{1,1}\right) = R_1 $ and $ R\left(K_{\mathcal{Q}_5(2)}^{1,3}\right) = (R_1^2 + R_2^2)R_3 + R_6^2. $

Using the symmetrization map $\Lambda$ in $\eqref{eq:symme}$, the corresponding quantum case has the following integrals \begin{align*}
    \hat{A}_1 & = X_5, \text{ } \hat{A}_2 = X_6, \text{ } \hat{A}_3 = -X_4 + X_2 X_6 + X_3^2, \text{ } \hat{A}_4 =X_3 - \frac{1}{2}(X_2 X_5 -X_1 X_6) + X_3X_4,\\
    \hat{A}_5 & = X_4^2 +X_4 - X_1 X_5 ,\quad \text{ } \hat{A}_6 =  X_4 X_6 -X_3 X_5.
\end{align*} They form the quantum quadratic algebra $\hat{\mathcal{Q}}_5(2)$ with commutation relations $[ \hat{A}_1,\hat{A}_j] =0$ and \begin{align*}
    [\hat{A}_2,\hat{A}_3] & =  [\hat{A}_2,\hat{A}_5]=-2 \hat{A}_6, \text{ } [\hat{A}_2,\hat{A}_6] = -\hat{A}_1^2 - \hat{A}_2^2,  \\
 [\hat{A}_3,\hat{A}_6]& =    [\hat{A}_5,\hat{A}_6] = -2 \hat{A}_1\hat{A}_4 + \hat{A}_2\hat{A}_5 + 2 \hat{A}_6.
\end{align*}  By direct calculations, it can be shown that the Jacobi identity for the commutator relations holds. Hence $\hat{\mathcal{Q}}_5(2),$ with the commutator $[\cdot,\cdot]$, is an associative algebra.

\subsubsection{Symmetry algebra from $\mathfrak{a}_6$}
The commutant of $\mathfrak{a}_6$ induces the following PDE \begin{align*}
    \left\{p(x_1,\ldots,x_6),x_6\right\}_o = 2x_3 \dfrac{\partial p}{\partial x_1} +2 x_4 \dfrac{\partial p}{\partial x_2} - x_5 \dfrac{\partial p}{\partial x_3}  + x_6 \dfrac{\partial p}{\partial x_4} = 0.
\end{align*} Calculations show that the PDE above has the following linearly independent and indecomposable homogeneous polynomial solutions \begin{align*}
    A_1 &= x_5 ,\text{ } A_2 =  x_6, \text{ } A_3 =x_5 x_1 + x_3^2, \text{ } A_4 =  x_6 x_1 - x_2 x_5 - 2x_3 x_4 , \\
   A_5 & = x_2 x_6 - x_4^2, \text{ } \quad  A_6 = x_6 x_3 + x_4 x_5
\end{align*}   which  form a finite set $\textbf{Q}_2^{(6)}$.

All possible algebraic Hamiltonian is of the form \begin{align*}
    \mathcal{H}_6 = \alpha A_2 + \gamma_1 \mathcal{C}_1 + \gamma_2 \mathcal{C}_2,
\end{align*} where $\alpha = \Gamma_{000001} \in \mathbb{R}$ is an arbitrary coefficient.    Then $\{A_2,A_j\} =0$ for all $1 \leq j \leq 6$ and the non-zero Poisson brackets are given by \begin{align*}
    \{A_1,A_3\} & = - 2A_6 = -  \{A_1,A_5\}     , \text{ } \{A_1,A_6\}   =-A_1^2 - A_2^2 \\
      \{A_3,A_6\} & =   -\{A_5,A_6\}   = A_2 A_4 + 2 A_1A_5.
\end{align*}After checking the Jacobi identity vanished for all the generators in  $\textbf{Q}_2^{(6)}$, we verify that the above relations form the quadratic Poisson algebra $\mathcal{Q}_6(2) = \left(\textbf{Alg} \langle \textbf{Q}_2^{(6)} \rangle ,\{\cdot,\cdot\}\right) $. Moreover, since $A_4 = \mathcal{C}_1 $, $S_6 = \big\{\mathcal{H}_6 = \alpha A_2 + \gamma_1 A_4, A_1,A_3,A_5,A_6 \big\}$ forms a superintegrable system. This algebra admits the following functionally independent Casimir functions \begin{align*}
        K_{\mathcal{Q}_6(2)}^{1,1} = A_2  , \text{  }  K_{\mathcal{Q}_6(2)}^{1,2} =   A_4, \text{ } K_{\mathcal{Q}_6(2)}^{1,3} = A_5 + A_3, \text{  } K_{\mathcal{Q}_6(2)}^{3,1} = A_1^2 A_3-A_5 A_2^2-A_2 A_4 A_1+A_6^2.
\end{align*} Similar to the result presented in Subsection $\ref{4.1.5},$  the classical version of the realisation $ \eqref{eq:real}$ gives functionally independent Casimir $ R\left(  K_{\mathcal{Q}_6(2)}^{1,1}\right) = R_2 $ and $ R\left(K_{\mathcal{Q}_5(2)}^{1,3}\right) = (R_1^2 + R_2^2)R_3 + R_6^2. $

The corresponding quantum case has the integrals \begin{align*}
    \hat{A}_1 &= X_5, \text{ } \hat{A}_2 = X_6, \text{ } \hat{A}_3 = X_3^2 + X_1 X_5 -X_4, \\ 
    \hat{A}_4 &  =X_3 - \frac{X_2 X_5 -X_1 X_6 }{2} + X_3X_4, \\
    \hat{A}_5&  =X_4 + X_4^2 - X_2 X_6 X_3 , \text{ } \quad \hat{A}_6 = X_6 + X_4 X_5 .
\end{align*} Then $ [\hat{A}_1,\hat{A}_j] = 0$ for all $1 \leq j \leq 6$ and the non-zero commutators are \begin{align*}
    [\hat{A}_1,\hat{A}_3]  &  = [\hat{A}_1,\hat{A}_5]  = -2 \hat{A}_6,  \text{ } [\hat{A}_1,\hat{A}_6] = -A_1^2 - A_2^2,   \\
    [\hat{A}_3,\hat{A}_6] & =    [\hat{A}_5,\hat{A}_6]   = 2\hat{A}_2 \hat{A}_4 +  \hat{A}_1\hat{A}_5 + 2\hat{A}_6 
\end{align*}  We can check that the Jacobi identity holds for the commutator relations, and hence $\hat{\mathcal{Q}}_6(2)$ is associative.

\subsubsection{Reduction chain $\mathfrak{a}_c \subset \mathfrak{c}(2)$}

The commutant of $\mathfrak{a}_c$ induces the following PDE \begin{align}
    \left\{p(x_1,\ldots,x_6),a x_3 + b x_4\right\}_o = a  \left\{p(x_1,\ldots,x_6),x_3 \right\}_o+ b\left\{p(x_1,\ldots,x_6), x_4\right\}_o= 0. \label{eq:4.7}
\end{align} Here $a = \frac{\cos c_1}{2} $ and $ b = \frac{\sin c_2}{2}$ are non-zero $0 < c_1 < \frac{\pi}{2}$ and $\frac{\pi}{2} < c_2 < \pi$. Calculations show that the PDE $\eqref{eq:4.7}$ has the following linearly independent and indecomposable homogeneous polynomial solutions as follows \begin{align*}
    A_1 &= x_3 ,\text{ } A_2 =  x_4, \text{ } A_3 =x_5 x_1 + x_2 x_6 , \text{ } A_4 =  x_6 x_1 - x_2 x_5. 
\end{align*}  A direct calculation shows that $\{A_i,A_j\} =0$ for all $1 \leq i,j \leq 4.$ Hence the polynomial algebra generated by $\textbf{Q}_c = \{A_1,A_2,A_3,A_4\}$ is Abelian. It is easy to check that all these generators are functionally independent.  By definition, a possible algebraic Hamiltonian is given by \begin{align*}
    \mathcal{H}_6 = \alpha A_1 + \beta A_2 + \gamma_1 \mathcal{C}_1  
\end{align*}   Here $\mathcal{C}_1 = A_1^2 - A_2^2 + A_3.$ Therefore, all the elements in $ \textbf{Q}_c$ defines an integrable system.




\subsection{reduction chains $\mathfrak{a}_{(ij)} \subset \mathfrak{c}(2)$}

Recall that the two-dimensional subalgebras of the conformal algebra $\mathfrak{c}(2)$ are classified as Abelian and non-Abelian, respectively. The explicit generators for each $\mathfrak{a}_{(ij)}$ are provided in $\eqref{eq:2da}$ and $\eqref{eq:2dc}$.  Let $\mathfrak{a}_{(ij)}^* $ be its dual space. Then it is straightforward from $\eqref{eq:1}$ that the algebraic Hamiltonian can be expressed in the following way \begin{align*}
    \mathcal{H}_1 = \alpha x y + \beta_1 x + \beta_2 y + \gamma_1 \mathcal{C}_1 + \gamma_2 \mathcal{C}_2,
\end{align*} where $\alpha,\beta_1,\beta_2 \in \mathbb{R}$ are arbitrary coefficients. It is clear that $\mathcal{H}_1$ is a quadratic integral.  We first consider the Abelian case.

\subsubsection{Polynomial algebra generated by MASAs}


\begin{lemma}
Let $\mathfrak{c}(2)$ be the $6$-dimensional conformal algebra with basis $\eqref{eq:real}$, and let $\mathfrak{m}$ be a MASA of $\mathfrak{c}(2)$. Then the commutant of $\mathfrak{m}$ is Abelian. That is, $\mathcal{C}_{\mathcal{U}(\mathfrak{c}(2))}(\mathfrak{m})$ is Abelian.  
\end{lemma}

 \begin{proof}
In Section $\ref{3},$ all the maximal Abelian subalgebras of $\mathfrak{c}(2)$ are, up to conjugate, $\mathfrak{a}_{(12)},\mathfrak{a}_{(56)} \cong \mathfrak{e}(1) \oplustilde \mathfrak{e}(1)$ and $\mathfrak{a}_{(34)} \cong \mathfrak{o}(2) \oplustilde \mathfrak{o}(1,1)$, where $\oplustilde$ defines a direct sum of the Lie algebra. We give the explicit formula for the integrals by finding the commutant $\mathcal{C}_{\mathcal{S}(\mathfrak{c}(2))}\left(\mathfrak{a}_{(12)}\right)$, and then calculate their Poisson brackets.  From the commutator relations, the action of the elements in $\mathfrak{a}_{(12)}$ on the generators of $\mathcal{S}(\mathfrak{c}(2))$ forms the following matrix \begin{align*}
    C_1 = \begin{pmatrix}
    0 &  0 & x_2 & -x_1 & -2x_4 & -2x_3\\
    0 & 0 &   -x_1 &  -2x_2 & 2x_3& -2 x_4
    \end{pmatrix}.
\end{align*} It follows that \begin{align*} 
     \tilde{X}_1(p) &  = x_2 \dfrac{\partial p}{\partial x_3} + x_1 \dfrac{\partial p}{\partial x_4}-2 x_4 \dfrac{\partial p}{\partial x_5} -2 x_3 \dfrac{\partial p}{\partial x_6} = 0;   \\
   \tilde{X}_2(p) &  =  x_2 \dfrac{\partial p}{\partial x_4}+2 x_3 \dfrac{\partial p}{\partial x_5} -2 x_4 \dfrac{\partial p}{\partial x_6}  = 0.    
\end{align*} 

Then it is easy to verify the following $4$ indecomposable homogeneous polynomial solutions \begin{align*}
C_1: \quad \begin{matrix}
  &  \varphi_1   = x_1 ,\quad \varphi_2  = x_2 ,\quad \varphi_3  = x_3^2-x_4^2+x_1 x_5+x_2 x_6, \\
  &   \varphi_4  = -2 x_3 x_4-x_2 x_5+x_1 x_6.
    \end{matrix}
\end{align*} Let $A_j^{(1)} = \varphi_j(x_1,\ldots,x_6)$ for $1 \leq j \leq 4.$ Simple calculation shows that $ \left\{A_j^{(1)},A_k^{(1)}\right\} = 0.$ Similarly, the actions of the elements in $\mathfrak{a}_{(34)} $ and $\mathfrak{a}_{(56)}$ on the generators of $\mathfrak{c}(2)$ form the matrices \begin{align*}
    C_2  = \begin{pmatrix}
    -x_2 & x_1 & 0 & 0 & - x_6 & -x_5 \\
    x_1 & x_2 &  0 & 0 & -x_5 & -x_6
    \end{pmatrix} \text{ and }    C_3  = \begin{pmatrix}
        2x_4 & -2x_3  & x_6 & x_5 & 0 & 0 \\
    2x_3 & 2x_4 & -x_5 & x_6 &  0 & 0
    \end{pmatrix}
\end{align*} Then solutions generated from the matrices $C_2$ and $C_3$ are \begin{align*}
  C_2: \quad   \varphi_1 & = x_3, \text{ } \varphi_2  = x_4, \text{ } \varphi_3  = x_1x_5 + x_2 x_6, \text{ } \varphi_4  = x_1x_6 - x_2 x_5 ; \\
  C_3:  \quad \varphi_1 & = x_5, \text{ } \varphi_2  = x_6, \text{ } \varphi_3  = 2x_3x_4 + x_2 x_5 - x_1 x_6, \\ \varphi_4 &  =x_3^2-x_4^2+x_1 x_5+x_2 x_6.
\end{align*}  Let $A_j^{(s)} = \varphi_j(x_1,\ldots,x_6)$ for $s = 2,3$. It is easy to check that the Poisson brackets of these solutions are zero. Hence, using the symmetrization isomorphism $\eqref{eq:symme},$ the symmetry algebras generated from the subalgebra $\mathfrak{a}_{(34)}$ and $ \mathfrak{a}_{(56)}$ are Abelian. Moreover, note that the Abelian polynomial algebras for $\mathfrak{a}_{(12)}$ and $\mathfrak{a}_{(56)}$ are given by $\mathcal{P}_1 =\mathfrak{a}_{(12)} \oplus \mathcal{Z} \left(\mathcal{U}(\mathfrak{c}(2))\right)$ and $\mathcal{P}_5 =\mathfrak{a}_{(56)} \oplus \mathcal{Z} \left(\mathcal{U}(\mathfrak{c}(2))\right)$ respectively.

In conclusion, all $\mathcal{C}_{\mathcal{U}(\mathfrak{c}(2))}(\mathfrak{a}_{ij})$ are Abelian on the basis of $\eqref{eq:real}.$ \end{proof}

\subsubsection{Algebraic Hamiltonians from 2D non-Abelian subalgebras} 

Suppose that $\mathfrak{a}  = \mathrm{span} \{X,Y: [X,Y]_o = X\}.$ By definition $\ref{H}$, the algebraic Hamiltonians in the dual space are \begin{align*}
    \mathcal{H}_2 = \alpha x y + \beta_1 x + \beta_2 y + \gamma_1 \mathcal{C}_1 + \gamma_2 \mathcal{C}_2, 
\end{align*} which satisfies $\{\mathcal{H}_2,p(x_1,\ldots,x_6)\} = 0$ for all $ p(x_1,\ldots,x_6) \in \mathcal{C}_{\mathcal{S}(\mathfrak{c}(2))}(\mathfrak{a}),$ where $\alpha,\beta_1,\beta_2,\gamma_1$ and $\gamma_2$ are constants. We will also provide a classical realisation for the Hamiltonian in each case via the momentum realisation in \eqref{eq:real}. Notice that, in what follows, all the non-Abelian two-dimensional subalgebra chains are equivalent to $\mathfrak{b}_0$ (Borel subalgebra in $\mathfrak{sl}(2)$). From $\eqref{eq:2dc},$ we consider the reduction chains $\mathfrak{a}_{(14)} \subset \mathfrak{c}(2),$ $\mathfrak{a}_{(24)} \subset \mathfrak{c}(2),$ $\mathfrak{a}_{(45)} \subset \mathfrak{c}(2) $ and $\mathfrak{a}_{(46)} \subset \mathfrak{c}(2).$ 

\vskip.2in
\noindent {\bf Case 1:}\begin{large}
    \textbf{Reduction Chain $\mathfrak{a}_{(14)} \subset \mathfrak{c}(2)$}
\end{large}

\quad

Starting with the subalgebra $\mathfrak{a}_{(14)}.$ The possible algebraic Hamiltonian is \begin{align}
    \mathcal{H} = \alpha x_1 x_4 + \beta_1 x_1 + \beta_4 x_4 + \gamma_1 \mathcal{C}_1 + \gamma_2 \mathcal{C}_2. \label{eq:eq14}
\end{align} Here, in this particular example, we have $\alpha = \Gamma_{100100}, \beta = \Gamma_{100000},\beta_4 = \Gamma_{000100} \in \mathbb{R}$. In order to facilitate a clearer understanding throughout this discussion, all the coefficients in Hamiltonians, in the following sections, have been substituted with the variables $\alpha$ and the respective $\beta_j, \beta_k$. Using the momentum realisation in $\eqref{eq:real}$, the Hamiltonian $\eqref{eq:eq14}$ is realized as \begin{align*}
    \title{\mathcal{H}}_{14} = \alpha x p_x^2 + \alpha y p_x p_y + (\beta_1 + \beta_2 x) p_x + \beta_2 y p_y.
\end{align*} 
For the system of PDEs $ \left\{p(x_1,\ldots,x_6),x_1\right\}_o = \left\{p(x_1,\ldots,x_6),x_4\right\}_o = 0$, it is sufficient to rewrite them as $C_{14} \cdot \left(\dfrac{\partial p}{\partial x_1},\ldots ,\dfrac{\partial p}{\partial x_6}\right)^T =0$, where \begin{align*}
    C_{14} = \begin{pmatrix}
        0& 0& x_2 & -x_1 & -2x_4 & -2x_3 \\
        x_1 & x_2 & 0& 0 & -x_5& -x_6
    \end{pmatrix}.
\end{align*} By using symbolic computing packages, we deduce that the linearly independent  and indecomposable homogeneous solutions of the PDEs above are \begin{align*}
    \varphi_1  & = x_1x_5  - x_4^2 , \text{ }\quad \varphi_2  =x_2 x_6 + x_3^2, \text{ } \quad \varphi_3  = x_1 x_6 -2x_3 x_4 - x_2 x_5 .
\end{align*} Notice that $\{\varphi_i,\varphi_j\}  = 0$ for all $1 \leq i, j \leq 4,$ which implies that the commutant $\mathcal{C}_{\mathcal{U}(\mathfrak{c}(2))}(\mathfrak{a}_{(14)})$ is Abelian.

\vskip.2in
\noindent {\bf Case 2:}\begin{large}
    \textbf{Reduction Chain $\mathfrak{a}_{(24)} \subset \mathfrak{c}(2)$}
\end{large}

\quad

Another example is the subalgebra $\mathfrak{a}_{(24)}$. The possible algebraic Hamiltonian is \begin{align}
    \mathcal{H}_{24} = \alpha x_2 x_4 + \beta_2 x_2 + \beta_4 x_4 + \gamma_1 \mathcal{C}_1 + \gamma_2 \mathcal{C}_2. \label{eq:eq24}
\end{align} In terms of the realisation in $\eqref{eq:real}$, this algebraic Hamiltonian is \begin{align*}
    \title{\mathcal{H}}_{24} = \left(\alpha  y + (\beta_1 + \beta_2 y) \right)p_y + \alpha y p_y^2  + \beta_2 x p_x.
\end{align*} The system of PDEs $\{p(x_1,\ldots,x_6),x_2\}_o = \{p(x_1,\ldots,x_6),x_4\}_o = 0$ admits the following coefficients matrix \begin{align*}
    C_{24} = \begin{pmatrix}
        0&0 &-x_1 &-x_2 & 2x_3& -2x_4\\
        x_1& x_2 & 0 & 0 & -x_5 & -x_6
    \end{pmatrix}.
\end{align*} We can show that linearly independent and indecomposable homogeneous polynomial solutions of the system are \begin{align*}
    \varphi_1  & =  x_2 x_6 - x_4^2, \text{ }\quad \varphi_2  = x_1 x_6 -2 x_3x_4 - x_2 x_5  , \text{ } \quad \varphi_3  = x_1 x_5 + x_3^2.
\end{align*}  Note that $\{\varphi_i,\varphi_j\} = 0$ for all $1 \leq i, j \leq 3 $. It follows that the commutant $\mathcal{C}_{\mathcal{U}(\mathfrak{c}(2))}(\mathfrak{a}_{(24)} )$ is Abelian. Let $A_j = \varphi_j(x_1,\ldots,x_6)$ for $1 \leq j \leq 3.$ Then the Hamiltonian $\mathcal{H}_{24}$ with integral $A_j$ forms an integrable system.

\vskip.2in
\noindent {\bf Case 3:}\begin{large}
    \textbf{Reduction Chain $\mathfrak{a}_{(45)} \subset \mathfrak{c}(2)$}
\end{large}

\quad

We now have a look at the commutant of the subalgebra $\mathfrak{a}_{(45)}$. The possible algebraic Hamiltonian is \begin{align}
    \mathcal{H}_{45} = \alpha x_4 x_5 + \beta_4 x_4 + \beta_5 x_5 + \gamma_1 \mathcal{C}_1 + \gamma_2 \mathcal{C}_2. \label{eq:eq45}
\end{align}  In terms of the realisation, the algebraic Hamiltonian $\eqref{eq:eq45}$ becomes \begin{align*}
    \title{\mathcal{H}}_{45}  & = \alpha  \left((x^2 - y^2)x p_x^2 + 2 xy^2 p_y^2  + (x^2 - y^2 + 2 x^2 y )y p_xp_y \right) \\
    &+ (\beta_1 x + \beta_2(x^2 - y^2)  + \beta_1 x) p_x + (2 \beta_2 x y + \beta_1 y) p_y.
\end{align*} The coefficient matrix for the system of PDEs $\{p(x_1,\ldots,x_6),x_5\}_o = \{p(x_1,\ldots,x_6),x_4\}_o = 0$  is given by \begin{align*}
    C_{45} = \begin{pmatrix}
        x_1& x_2 & 0 & 0 & -x_5 & -x_6\\
       2 x_4 & -2x_3 & x_6  & x_5 & 0 & 0 
    \end{pmatrix}.
\end{align*} We can show that indecomposable homogeneous polynomial solutions of the system $C_{45}$ are \begin{align*}
    \varphi_1  & = x_3^2 + x_2 x_6, \text{ }\quad \varphi_2  =-2 x_3 x_4 - x_2 x_5 + x_1 x_6, \text{ }\quad \varphi_3  = x_1 x_5 -x_4^2.  
\end{align*}  Notice that $\{\varphi_i,\varphi_j\}  = 0$ for all $1 \leq i, j \leq 3,$ which implies that the commutant $\mathcal{C}_{\mathcal{U}(\mathfrak{c}(2))}(\mathfrak{a}_{(45)} )$ is Abelian. 

\vskip.2in
\noindent {\bf Case 4:}\begin{large}
    \textbf{Reduction Chain $\mathfrak{a}_{(46)} \subset \mathfrak{c}(2)$}
\end{large}

\quad

Now, we consider the commutant of the subalgebra $\mathfrak{a}_{(46)}$.  The possible algebraic Hamiltonian is \begin{align}
    \mathcal{H}_{46} = \alpha x_4 x_6 + \beta_4 x_4 + \beta_6 x_6 + \gamma_1 \mathcal{C}_1 + \gamma_2 \mathcal{C}_2. \label{eq:eq46}
\end{align}  Moreover, the algebraic Hamiltonian $\eqref{eq:eq46}$ in terms of realisation $\eqref{eq:real}$ is \begin{align*}
    \title{\mathcal{H}}_{46} = \alpha  \left( 2 x^2 y p_x^2 + x(3y^2 -x^2) p_xp_y  + (y^2 - x^2) y p_y^2\right) + \left(\beta_1 x + 2 \beta_2 x y \right) p_x + (\beta_1 y + \beta_2(y^2 - x^2) ) p_y.
\end{align*} The system of PDEs $\{p(x_1,\ldots,x_6),x_6\}_o = \{p(x_1,\ldots,x_6),x_4\}_o = 0$ has the following coefficient matrix \begin{align*}
    C_{46} = \begin{pmatrix}
       x_1& x_2 & 0 & 0 & -x_5 & -x_6\\
       2 x_3 & 2 x_4 & -x_5 & x_6 & 0 & 0
    \end{pmatrix}.
\end{align*} By using symbolic computing packages, the linearly independent and indecomposable solutions of the system $C_{46}$ are \begin{align*}
    \varphi_1  & = x_2x_6 - x_4^2, \text{ } \quad \varphi_2  = -2 x_3x_4 -x_2 x_5 + x_1 x_6, \text{ } \quad \varphi_3  = x_1 x_5 + x_3^2.  
\end{align*}   Notice that $\{\varphi_i,\varphi_j\} = 0$ for all $1 \leq i, j \leq 3$, which implies that the commutant $\mathcal{C}_{\mathcal{U}(\mathfrak{c}(2))}(\mathfrak{a}_{(46)} )$ is Abelian. Let $A_j = \varphi_j$ for all $1 \leq j \leq 3.$ With the Hamiltonian $\mathcal{H}$ defined in $ \eqref{eq:eq46}$, we obtain an integrable system $\{\mathcal{H}_{46},A_j\} =0.$   

Using the calculations for all of the above cases, we can conclude the following.  

\begin{proposition}
    Let $\mathfrak{c}(2)$ be a $6$-dimensional Lie algebra on the basis $\eqref{eq:real}$. Then the integrals that form the Abelian polynomial algebra generated from the subalgebra chain $\mathfrak{c}(1) \subset \mathfrak{c}(2)$ are \begin{align}
    \hat{A}  = \sum_{1 \leq k,j \leq 6, i_k + i_j = 2}\Gamma_{i_k,i_j} X_k^{i_k} X_j^{i_j} \in \hat{\textbf{q}}_2. \label{eq:qua}
\end{align}   Moreover, the degrees of all the homogeneous polynomial solutions are even.
\end{proposition}
\begin{proof}

Observe that the expression $\eqref{eq:qua}$ and the connection between the commutants can be inferred from the computations in the preceding sections. We will now demonstrate that these commutants exclude any integrals of odd degrees.

We use the induction method for the degrees of $P_h ,$ where $h = \deg P (X_1,\ldots,X_6)$. Let $h= 1.$ Assume that there exists a degree one polynomial such that $[P_1 ,\mathfrak{a}_{(ij)}] = 0.$ Take $P_1  = \sum_{i=1}^6 a_iX_i,$ where $a_i \neq 0$ are some constants for all $i.$ A simple calculation shows that \begin{align}
    [P_1 ,   X_1]=  [P_1 ,    X_4] = 0 \Rightarrow a_i = 0 \text{ for all } i. \label{eq:ii}
\end{align} Similar argument works for $\mathfrak{a}_{(24)}, \mathfrak{a}_{(45)}$ and $\mathfrak{a}_{(46)}$. Hence $P_1 = 0, $ which is a contradiction. Next, assume that $h= 3.$ The possible combinations are $P_3   =P_1 P_2 $ or $P_3 = P_1^3 $, where $P_1 \in \hat{\textbf{q}}_1$ and $P_2 \in \hat{\textbf{q}}_2.$ From the argument above, it is clear that the first case fails. Hence all $P_3 $ in the solution sets are linearly independent. Let $P_3  = X_1^{i_1}\cdots X_6^{i_6} \in \hat{\textbf{q}}_3$. Then \begin{align}
 0 =   [X_1^{i_1}\cdots X_6^{i_6}, a X_j + b X_4] \quad \Rightarrow \quad  [X_1^{i_1}\cdots X_6^{i_6},   X_j] = [X_1^{i_1}\cdots X_6^{i_6},   X_4] = 0 , \label{eq:fail1}
\end{align} where $j= 1,2,5,6$ and $a,b$ are non-zero constants. Using \cite[Lemma 1]{marquette2023infinite}, a simple calculation shows that \begin{align*}
     [X_1^{i_1}\cdots X_6^{i_6},  X_4] = \left(\sum_{j=1}^6i_j \lambda_j\right) X_1^{i_1}\cdots X_6^{i_6} = \left(i_1 + i_2 -i_5 -i_6 \right)P_3 ,
\end{align*} where $\lambda_1 = 1=\lambda_2,\lambda_3 = \lambda_4 = 0 $ and $\lambda_5 = \lambda_6 = -1.$  We then deduce another constraint $i_1 + i_2 = i_5 + i_6.$ Then the constraint in the degree of the monomials from the set $\hat{\textbf{q}}_3 $ becomes \begin{align}
    2 \left(i_1 + i_2\right) 
 +i_3 + i_4 = 3 \text{ or }    2 \left(i_5 + i_6\right) 
 +i_3 + i_4 = 3\label{eq:43}
\end{align}   Moreover, a routine computation shows that \begin{align}
 \nonumber
     [X_1^{i_1}\cdots X_6^{i_6},  X_1] = & \, i_3 X_1^{i_1} X_2^{i_2 + 1} X_3^{i_3-1} \cdots X_6^{i_6} - i_4  X_1^{i_1+1} X_2^{i_2} \cdots X_4^{i_4-1}X_5^{i_5}  X_6^{i_6} \\
      & - 2 i_5 X_1^{i_1} \cdots X_4^{i_4+1}X_5^{i_5-1}  X_6^{i_6} - 2i_6 X_1^{i_1} \cdots X_3^{i_3+1 }X_4^{i_4} X_5^{i_5}  X_6^{i_6-1}   = 0 \label{eq:45} \\
      \nonumber
      [X_1^{i_1}\cdots X_6^{i_6},  X_2] = & \, -i_3 X_1^{i_1+1} X_2^{i_2 } X_3^{i_3-1} \cdots X_6^{i_6} - i_4  X_1^{i_1} X_2^{i_2+1} \cdots X_4^{i_4-1}X_5^{i_5}  X_6^{i_6} \\
      & + 2 i_5 X_1^{i_1} \cdots X_4^{i_4+1}X_5^{i_5-1}  X_6^{i_6} - 2i_6 X_1^{i_1} \cdots   X_4^{i_4+1} X_5^{i_5}  X_6^{i_6-1}= 0 \label{eq:46} \\
      \nonumber
      [X_1^{i_1}\cdots X_6^{i_6},  X_5] = & \, 2i_1 X_1^{i_1-1} \cdots X_4^{i_4+1} X_5^{i_5 } X_6^{i_6} - 2i_2  X_1^{i_1} X_2^{i_2-1} X_3^{i_3+1} \cdots X_4^{i_4}X_5^{i_5}  X_6^{i_6} \\
      & +   i_3 X_1^{i_1} \cdots X_3^{i_3-1}X_5^{i_5}  X_6^{i_6+1} + i_4 X_1^{i_1} \cdots   X_4^{i_4-1} X_5^{i_5+1}  X_6^{i_6}= 0 \label{eq:47}\\
     \nonumber 
        [X_1^{i_1}\cdots X_6^{i_6},  X_6] = & \, 2i_1 X_1^{i_1-1} \cdots X_3^{i_3+1} X_4^{i_4+1} X_5^{i_5 } X_6^{i_6} + 2i_2  X_1^{i_1} X_2^{i_2-1}   \cdots X_4^{i_4+1}X_5^{i_5}  X_6^{i_6} \\
        & -   i_3 X_1^{i_1} \cdots X_3^{i_3-1}X_4^{i_4}X_5^{i_5+1}  X_6^{i_6} + i_4 X_1^{i_1} \cdots   X_4^{i_4-1} X_5^{i_5}  X_6^{i_6+1}= 0 \label{eq:48}
 \end{align} On the one hand, we first assume that $P_3 \in \mathcal{C}_{\mathcal{U}(\mathfrak{c}(2))}\big(\mathfrak{a}_{(24)} \big)$ or $\mathcal{C}_{\mathcal{U}(\mathfrak{c}(2))}\big(\mathfrak{a}_{(24)}\big).$ Then $[P_3,X_1] = [P_3,X_2] =0.$ From $\eqref{eq:45}$ and $\eqref{eq:46},$ we obtain $i_3 = i_4 = i_5 = i_6 = 0.$ Together with the constraint $\eqref{eq:43},$ $i_1 + i_2 = \frac{3}{2},$ which is a contradiction to $i_j \in \mathbb{N}_0,$ i.e., the closure of $\mathbb{N}_0.$ On the other hands, if $P_3 \in \mathcal{C}_{\mathcal{U}(\mathfrak{c}(2))}\big(\mathfrak{a}_{(45)} \big)$ or $\mathcal{C}_{\mathcal{U}(\mathfrak{c}(2))}\big(\mathfrak{a}_{(46)}\big)$, from $\eqref{eq:47}$ and $\eqref{eq:48}$, we deduce $i_1 = i_2 = i_3 = i_4 = 0.$  Applying the constraint $\eqref{eq:43}$ once more results in the same contradiction. Therefore, $P_3$ is not an element of $\hat{\textbf{q}}_3.$
 
 





Using logic similar to the constraint $2\left(i_1 + i_2 \right) + i_3 + i_4 = 2k+1$, for any $k > 1$, together with the identities in $\eqref{eq:45}$ to $\eqref{eq:48}$, it follows that there are no homogeneous polynomial solutions of odd degrees.
\end{proof}

\subsection{Three dimensional subalgebras $\mathfrak{a}_{(ijk)}$ and $\mathfrak{a}_{(i,j,k)}$ with $1 \leq i \neq j \neq k \leq 6$}
Next, we examine all three-dimensional subalgebras $\mathfrak{a}_{(ijk)}$ and $\mathfrak{a}_{(i,j,k)}$ for $1 \leq i \neq j \neq k \leq 6$ within $\mathfrak{c}(2)$, along with their explicit generators listed in $\eqref{eq:3dc}$. From Section $\ref{3}$, we know that they are classified as $\mathfrak{sl}(2)$, $\mathfrak{e}(2)$, $\mathfrak{su}(2)$, $\mathfrak{su}(1,1)$, $\mathfrak{n}_{3,3}$ and $\mathfrak{iso}(2)$. In the subsequent discussion, we will demonstrate that the symmetry algebras produced by these $3$-dimensional subalgebras are Abelian.

\subsubsection{Symmetry algebra from the subaglebra $\mathfrak{sl}(2)$}

From the above section, we know that, under the basis $\eqref{eq:real},$ the Lie subalgebra $\mathfrak{a}_{(145)}$ is conjugate to $\mathfrak{a}_{(246)}.$ Consider the reduction chains $\mathfrak{a}_{(145)} \subset \mathfrak{c}(2)$ and $\mathfrak{a}_{(246)} \subset \mathfrak{c}(2)$. We have the systems of PDEs  \begin{align*}
    C_{145}:& \quad \{p(x_1,\ldots,x_6), x_1\}_o =   \{p(x_1,\ldots,x_6), x_4\}_o  = \{p(x_1,\ldots,x_6), x_5\}_o = 0; \\
    C_{246} :& \quad \{p(x_1,\ldots,x_6), x_2\}_o =  \{p(x_1,\ldots,x_6), x_4\}_o  = \{p(x_1,\ldots,x_6), x_6\}_o = 0.
\end{align*}  The corresponding coefficient matrices  of the PDEs are  \begin{align*}
    C_{145} := \begin{pmatrix}
        0 & 0 & x_2 & -x_1 & -2 x_4 & -2 x_3 \\
        x_1 &  x_2 &  0 & 0 & -x_5 & -x_6 \\
        2 x_4 & -2 x_3 & x_6 & x_5 & 0 & 0 
    \end{pmatrix}, \text{ }\quad C_{246} := \begin{pmatrix}
        0 & 0 & -x_1 & -x_2 & 2 x_3 & -2 x_4 \\
        x_1 &  x_2 &  0 & 0 & -x_5 & -x_6 \\
        2 x_3 &  2 x_4 & -x_6 & x_6 & 0 & 0 
    \end{pmatrix} .
\end{align*} A simple calculation shows that each of them has $3$ indecomposable polynomial solutions \begin{align*}
    C_{145}: \varphi_1  & = x_1 x_5 - x_4^2, \text{ } \varphi_2  = x_2 x_6 + x_3^2, \text{ } \varphi_3  = x_2 x_5- x_1 x_6 + 2 x_3 x_4; \\
    C_{246} : \varphi_1  & = x_2 x_6 - x_4^2, \text{ } \varphi_2  = x_5 x_1  + x_3^2, \text{ } \varphi_3  = x_6 x_1 - x_2 x_5 - 2 x_3 x_4.
\end{align*}  Let $A_j = \varphi_j(x_1,\ldots,x_6)$ for all $1 \leq j \leq 3$ in both cases. We obtain that $\{A_j,A_k\} = 0$ for all $1 \leq j, k \leq 3,$ and therefore they form Abelian symmetry algebras.   The possible algebraic Hamiltonian for those integrable systems are \begin{align*}
      &  \text{ } \mathcal{H}_{145} = \alpha_{14} x_1x_4 + \alpha_{15}x_1x_5 + \alpha_{45}x_4x_5 + \sum_{j=1,j\neq 3}^5\beta_j x_j + \gamma_1 \mathcal{C}_1 + \gamma_2 \mathcal{C}_2 \\
     &  \text{ } \mathcal{H}_{246} = \alpha_{24} x_2x_4 + \alpha_{26}x_2x_6 + \alpha_{46}x_4x_6 + \sum_{j=3,j\neq 3,5}^3\beta_jx_j + \gamma_1 \mathcal{C}_1 + \gamma_2 \mathcal{C}_2.
\end{align*}  Hence, the Hamiltonian with integrals $A_j$ form integrable systems.

\subsubsection{Symmetry algebra from the subaglebras $\mathfrak{e}(2)$}
Under the basis $\eqref{eq:real},$ we have the subalgebra $\mathfrak{a}_{(123)}$ and $\mathfrak{a}_{(563)}$. The systems of PDEs are given by \begin{align*}
   C_{123}:& \quad \{p(x_1,\ldots,x_6), x_1\}_o =   \{p(x_1,\ldots,x_6), x_2\}_o  = \{p(x_1,\ldots,x_6), x_3\}_o = 0 \\
    C_{563} :& \quad \{p(x_1,\ldots,x_6), x_5\}_o =  \{p(x_1,\ldots,x_6), x_6\}_o  = \{p(x_1,\ldots,x_6), x_3\}_o = 0.
\end{align*} They have the following coefficient matrices \begin{align*}
    C_{123}:  \begin{pmatrix}
        0 & 0 & x_2 & -x_1 & -2 x_4 & -2 x_3 \\
       0 &  0 &  -x_1 & -x_2 & 2 x_3 & -2 x_4 \\
       -x_2 &  x_1 & 0 &0  & -x_6 & x_5 
    \end{pmatrix} \text{ and } C_{563} = \begin{pmatrix}
           -x_2 &  x_1 & 0 &0  & -x_6 & x_5  \\
        2 x_4 & -2 x_3 & x_6 & x_5 & 0 & 0  \\
        2 x_3 &  2 x_4 & -x_6 & x_6 & 0 & 0 
    \end{pmatrix} .
\end{align*} After calculations, all the indecomposable homogeneous polynomial solutions are\begin{align*}
    C_{123} : & ~\varphi_1  = x_1^2 + x_2^2, \text{ } \quad\varphi_2  = x_3^2 - x_4^2 + x_1 x_5 + x_2 x_6 ,\text{ } \quad\varphi_3  = -2x_3x_4 - x_2 x_5 + x_1 x_6 \\
    C_{563} :& ~\varphi_1' = x_5^2 + x_6^2, \text{ } \quad\varphi_2' = x_3^2 - x_4^2 + x_1 x_5 + x_2 x_6 ,\text{ } \quad \varphi_3' =  2x_3x_4 + x_2 x_5-x_1 x_6
\end{align*} It can be shown that $\{\varphi_j ,\varphi_k \} =\{\varphi_j' ,\varphi_k' \} = 0$ for all $1 \leq j,k \leq 3.$ The possible algebraic Hamiltonians for these systems are \begin{align*}
     &  \text{ } \mathcal{H}_{(123)} = \alpha_{12}x_1x_2 + \alpha_{13}x_1x_3 + \alpha_{23}x_2x_3 + \sum_{j=1}^3\beta_jx_j + \gamma_1 \mathcal{C}_1 + \gamma_2 \mathcal{C}_2 \\
      &  \text{ } \mathcal{H}_{(563)} = \alpha_{35} x_3x_5 + \alpha_{36}x_3x_6 + \alpha_{56}x_5x_6 + \sum_{j=3,j\neq 4}^3\beta_jx_j + \gamma_1 \mathcal{C}_1 + \gamma_2 \mathcal{C}_2.
\end{align*}  Hence, the Hamiltonians with integrals $\varphi_j$ and $\varphi_j'$, respectively, form integrable systems.

\subsubsection{Reduction chain $\mathfrak{su}(2) \subset \mathfrak{c}(2)$ and $\mathfrak{su}(1,1) \subset \mathfrak{c}(2)$}
Under the basis $\eqref{eq:real},$ we have the subalgebras $\mathfrak{a}_{(3,6,5)}^\pm$. We can consider the following system of PDEs \begin{align}
\nonumber
  C_a: & \quad   \{p(x_1,\ldots,x_6), x_3\}_o  = \{p(x_1,\ldots,x_6),x_2 + x_6\}_o =  \{p(x_1,\ldots,x_6),x_1 +x_5\}_o = 0; \\  
  C_b: & \quad      \{p(x_1,\ldots,x_6), x_4\}_o  = \{p(x_1,\ldots,x_6),x_2 - x_6\}_o =  \{p(x_1,\ldots,x_6),x_1 -x_5\}_o = 0.\label{eq:sys}
\end{align}  Then the linearly independent and indecomposable homogeneous polynomial solutions for the system $\eqref{eq:sys}$ are \begin{align*}
    C_a:  \varphi_1  &= x_1^2 + x_2^2 + 2 x_3^2 + 2 x_4^2 + x_5^2 + x_6^2, \text{ } \quad\varphi_2  = x_3^2  - x_4^2 + x_1 x_5 + x_2 x_6, \\ \varphi_3  &  = -2 x_3 x_4 - x_2 x_5 + x_1 x_6 . \\
    C_b: \varphi_1   &=  x_1^2 + x_2^2 - 2 x_3^2 - 2 x_4^2 + x_5^2 + x_6^2, \text{ } \varphi_2  = x_3^2  - x_4^2 + x_1 x_5 + x_2 x_6 , \\
    \varphi_3  &=  -2 x_3 x_4 - x_2 x_5 + x_1 x_6  
\end{align*} Notice that $\varphi_2$ and $\varphi_3$ are Casimir invariants of $\mathfrak{c}(2) $ in both cases. A routine calculation shows that $ \{\varphi_i ,\varphi_j \} = 0$ for all $ 1 \leq i,j \leq 3.$ Hence the symmetry algebra generated from both $\mathfrak{su}(2)$ and $\mathfrak{su}(1,1)$ are Abelian.

\subsubsection{reduction chains $\mathfrak{n}_{3,3} \subset \mathfrak{c}(2)$ and $\mathfrak{iso}(2) \subset 
\mathfrak{c}(2)$}

Solving the following system of PDEs\begin{align*}
    C_{12c} :& \quad \{p(x_1,\ldots,x_6),x_1\}_o = \{p(x_1,\ldots,x_6),x_2\}_o = \{p(x_1,\ldots,x_6),x_c\}_o = 0 \\
    C_{456} :& \quad \{p(x_1,\ldots,x_6),x_4\}_o = \{p(x_1,\ldots,x_6),x_5\}_o = \{p(x_1,\ldots,x_6),x_6\}_o = 0
\end{align*} with the coefficients matrices \begin{align*}
    C_{12c} :\quad  \begin{pmatrix}
        0 & 0 & x_2 & -x_1 & - 2 x_4 & -2 x_3 \\
        0 & 0& -x_1 & -x_2 & 2 x_3  & - 2x_4  \\
         x_1 & x_2 & 0 & 0 & -x_5 & -x_6 \\
        x_1 & x_2 & 0 & 0& -x_5 & -x_6
    \end{pmatrix}, \text{ }\quad C_{456} :\quad   \begin{pmatrix}
        x_1 & x_2 & 0 & 0& -x_5 & -x_6 \\
        2 x_4 & -2 x_3 & x_6 & x_5& 0& 0 \\
        2 x_3 & 2x_4 & -x_5 & x_6 & 0 & 0
    \end{pmatrix},
\end{align*}  where $x_c =  \frac{\cos c_1}{2} x_3  + \frac{\sin c_2}{2}x_4$, $0 < c_1 < \frac{\pi}{2}$ and $\frac{\pi}{2} < c_2 < \pi,$  we deduce that \begin{align}
    \varphi_1  &= x_3^2 - x_4^2 + x_1 x_5 + x_2 x_6 , \text{ } \quad \varphi_2  = x_1 x_6 -x_2 x_5 - 2 x_3 x_4   \label{eq:sam}
\end{align} are linearly independent and indecomposable polynomial solutions. Utilizing computational tools, it can be demonstrated that other polynomial solutions take the form $\varphi = \sum_{m,n} \Gamma_{m,n} \varphi_1^m \varphi_2^n$, where $\Gamma_{m,n}$ are constants. It is straightforward to verify that $\{\varphi_1, \varphi_2\} = \{\varphi_i, \varphi\} = 0$ for all $i = 1, 2$. Thus, the polynomial algebras are Abelian, and the integrable systems related to these reduction chains possess the following algebraic Hamiltonians  \begin{align*}
    \mathcal{H}_{12c} & = \alpha_{12} x_1x_2 + \alpha_{1c}x_1x_c + \alpha_{2c} x_2x_c + \beta_1 x_1 + \beta_2 x_2 + \beta_4 x_c + \gamma_1 \mathcal{C}_1  +\gamma_2 \mathcal{C}_2, \\
    \mathcal{H}_{456} & = \alpha_{45} x_4x_5 + \alpha_{46}x_4x_6 + \alpha_{56} x_5x_6 + \beta_4 x_4 + \beta_5 x_5 + \beta_6 x_6 + \gamma_1 \mathcal{C}_1  +\gamma_2 \mathcal{C}_2,
\end{align*} 
respectively. Here, $\alpha_{jk},\beta_i$ are constants for all $i,j,k$.

  
\subsection{Reduction chain on the Borel subalgebra}
\label{3.8}

Using the basis $\eqref{eq:real},$ the Borel subalgebra is given by $\mathfrak{b} = \mathrm{span}\{X_3,X_4,X_5,X_6\}.$ The system of PDEs is given by  \begin{align*}
    \{p(x_1,\ldots,x_6),x_3\}_o = \{p(x_1,\ldots,x_6),x_4\}_o=\{p(x_1,\ldots,x_6),x_5\}_o=\{p(x_1,\ldots,x_6),x_6\}_o =0,
\end{align*} with the coefficients matrix \begin{align*}
    C = \begin{pmatrix}
       -x_2 & x_1 & 0 & 0  & -x_6 & x_5 \\
          x_1 & x_2 & 0 & 0 & -x_5 & -x_6 \\
             2 x_4 & -2x_3 & x_6 & x_5 & 0 & 0 \\
             2 x_3 &  2x_4 & -x_5 & x_6 & 0 & 0 \\   
    \end{pmatrix}.
\end{align*} By using symbolic computing packages, the linearly independent and indecomposable solutions to the system of PDEs above are \begin{align*}
    \varphi_1   =  - x_6 x_1 + x_2 x_5 + 2 x_3 x_4 , \text{ }\quad \varphi_2  = x_1 x_5 + x_2 x_6 + x_3^2 - x_4^2
\end{align*} Since $\{\varphi_1 , \varphi_2 \} = 0,$ it follows that the Poisson centralizer $\mathcal{C}_{\mathcal{S}(\mathfrak{c}(2))}(\mathfrak{b})$ is Abelian. Possible algebraic Hamiltonian can be written in the form  \begin{align*}
    \mathcal{H} =\sum_{3 \leq j< k < l \leq 6} \delta_{j,k,l} x_jx_k x_l +  \sum_{3 \leq j< k \leq 6} \alpha_{j,k} x_jx_k + \sum_{j=3}^6 \beta_j x_j + \gamma_1 \mathcal{C}_1 + \gamma_2 \mathcal{C}_2.
\end{align*} Here $\delta_{j,k,l},\alpha_{j,k}$ and $\beta_j$ are constants.
 Hence, using the symmetrization isomorphism $\eqref{eq:symme}$, the commutant $\mathcal{C}_{\mathcal{U}(\mathfrak{c}(2))}(\mathfrak{b})$ is also Abelian. 

 \subsection{Superintegrable Hamiltonians from  Casimir invariants}
 \label{4.5}

The key idea in \cite{MR3988021} is to construct a Hamiltonian from the Casimir invariant of a subalgebra of $\mathfrak{c}(2)$ in the basis $\eqref{eq:real}.$ Now, from the construction in Section $\ref{2}$, we know that there exists a finitely-generated Poisson algebra $\mathcal{Q}(d) \subset \mathbb{C}[\mathfrak{c}^*(2)] \cong \mathcal{S}(\mathfrak{c}(2))$ such that $\mathcal{Q}(d)$ corresponds an algebraic  superintegrable system on $\mathfrak{c}(2).$ In this Section \ref{4.5}, we have a close look at the centralizer $\mathcal{C}_{\mathcal{S}(\mathfrak{g})}(K_{\mathfrak{a}})$, where $K_{\mathfrak{a}}$ is the Casimir invariant of the subalgebra $\mathfrak{a}^*$.  We will find all indecomposable homogeneous polynomials $p(x_1,\ldots,x_6) \in \mathcal{S}(\mathfrak{c}(2))$ such that $\{p(x_1,\ldots,x_6),K_{\mathfrak{a}}\}_o=0.$ Then $K_{\mathfrak{a}}$ defines a new Hamiltonian. 

In what follows, we consider Casimirs of all subaglebras and construct the corresponding superintegrable Hamiltonians. Since the Casimir invariants of $2$-dimensional subalgebras of $\mathfrak{c}(2)$ are constant or zero, we will focus on $3$-dimensional subalgebras. That is, $\mathfrak{sl}(2), \mathfrak{so}(3)$ or $\mathfrak{e}(2).$

\subsubsection{Subalgebra $\mathfrak{e} (2) $}
\label{4.5.1}

The Euclidean algebra $\mathfrak{e}(2)$, up to conjugacy, can be associated with $\mathfrak{a}_{(123)}$ and $\mathfrak{a}_{(563)}$ as in $\eqref{eq:3dc}$. In this Subsection \label{4.5.1}, we examine each case individually. The Casimir operator for $\mathfrak{a}_{(123)}^*$ is expressed as $K_{123} = x_1^2 + x_2^2$. Using symbolic computing packages, there are $5$ indecomposable and linearly independent polynomials such that $\left\{p(x_1,\ldots,x_6),K_{123}\right\}_o = 0$, which are given by \begin{align*}
    \varphi_1  &= x_1,\quad \varphi_2  = x_2 , \quad \varphi_3  = x_3,  \\
    \varphi_4  & = x_1 x_5 + x_2 x_6 -x_4^2 ,\quad \varphi_5  = x_1 x_6  - x_2 x_5 - 2 x_3 x_4.
\end{align*} Let $\varphi_j(x_1,\ldots,x_6) = A_j$ for all $1 \leq j \leq 5.$ They form a closed quadratic algebra $\mathcal{Q}_{\mathfrak{e} (2)}(2)$ under the non-zero Poisson brackets as follows \begin{align}
\nonumber
    \{A_1,A_3\} = & \,  -A_2, \text{ } \{A_1,A_4\} =    2 A_2A_3, \\
    \{A_2,A_3\} = & \, A_1, \text{ } \{A_2,A_4\}= -2A_1A_3. \label{eq:A_1}
\end{align} By a simple calculation, we can show that $\left\{A_i,\left\{A_j,A_k\right\}\right\} +\left\{A_j,\left\{A_k,A_i\right\}\right\}+\left\{A_k,\left\{A_i,A_j\right\}\right\}= 0$ holds for all $1 \leq i,j,k \leq 5$. Hence $\mathcal{Q}_{\mathfrak{e} (2)}(2)$ forms a Poisson quadratic algebra. A potential free superintegrable system is then defined by $S_3=\{K_{123},A_3,A_4,A_5\}$. Taking \begin{align}
    K_{\mathcal{Q}_{\mathfrak{e} (2)}(2)}^h =   \sum_{i_1 + \cdots + i_5 = h} \Gamma_{i_1,\ldots,i_5} A_1^{i_1} \cdots A_5^{i_5} \in  \mathcal{U}_h\left(\mathcal{Q}_{\mathfrak{e} (2)}(2)\right) .
\end{align} Through straightforward computation, we determine that \begin{align*}
    K_{\mathcal{Q}_{\mathfrak{e} (2)}(2)}^1 = & a_1 A_5 , \\
    K_{\mathcal{Q}_{\mathfrak{e} (2)}(2)}^2 = & a_1 A_5 + a_2 A_5^2 + a_3\left(A_3^2 + A_4\right) + a_4\left( A_1^2 + A_2^2\right), \\
   K_{\mathcal{Q}_{\mathfrak{e} (2)}(2)}^3 = &  K_{\mathcal{Q}_{\mathfrak{e} (2)}(2)}^2 +    a_5 \left(A_3^2+A_4\right) A_5 + a_6\left(A_1^2+A_2^2\right) A_5+ a_7 A_5^3 \\
   K_{\mathcal{Q}_{\mathfrak{e} (2)}(2)}^4 =&  K_{\mathcal{Q}_{\mathfrak{e} (2)}(2)}^3 +  a_8  \left(A_3^2+A_4\right)^2+  a_9 \left(A_1^2+A_2^2\right) \left(A_3^2+A_4\right) + a_{10}\left(A_3^2+A_4\right) A_5^2  \\
   & + a_{11}\left(A_1^2+A_2^2\right)^2+a_{12}\left(A_1^2+A_2^2\right) A_5^2+ a_{13}A_5^4\\
   & \vdots 
\end{align*} 
Here, as in Subsection $\ref{4.1.1},$ we have used $a_j $ to denote the real coefficients $\Gamma_{i_1,\ldots,i_5}$ for convenience. Each coefficient corresponds to a linearly independent Casimir. Using induction, we determine that for every $k$, there exists a polynomial $t_k^1$ such that \begin{align*} K_{\mathcal{Q}_{\mathfrak{e}(2)}(2)}^k = t_k^1 \left(A_5, A_3^2 + A_4, A_1^2 + A_2^2\right) + \sum_i^k \left( K_{\mathcal{Q}_{\mathfrak{e}(2)}(2)}^1\right)^i. \end{align*} Hence the functionally independent Casimir functions are \begin{align*}
    K_{\mathcal{Q}_{\mathfrak{e} (2)}(2)}^{1,1} =  A_5, \text{ }  K_{\mathcal{Q}_{\mathfrak{e} (2)}(2)}^{2,1} = A_3^2 +A_4, \text{ } K_{\mathcal{Q}_{\mathfrak{e} (2)}(2)}^{2,2}  = A_1^2 + A_2^2 = K_{123}. 
\end{align*}

After the quantization, the quadratic Poisson algebra becomes the quadratic commutator algebra $\hat{\mathcal{Q}}_{\mathfrak{e} (2)}(2)$ with non-zero commutators as follows \begin{align}
    [\hat{A}_1,\hat{A}_3] = -\hat{A}_2, \text{ } [\hat{A}_1,\hat{A}_4] =  2 \{\hat{A}_2,\hat{A}_3\}_a, \text{ } [\hat{A}_2,\hat{A}_3] = \hat{A}_1,\text{ } [\hat{A}_2,\hat{A}_4] = - 2 \{\hat{A}_1,\hat{A}_3\}_a. \label{eq:A_1'}
\end{align}  By verifying the Jacobi identity, we can show that $\hat{\mathcal{Q}}_{\mathfrak{e} (2)}(2)$ is an associative algebra. As we can observe from the realisation above, $\hat{A}_5$ is in the center of the commutant $\mathcal{C}_{\mathcal{U}(\mathfrak{c} (2))}(\mathfrak{a}_{(123)})$.

The Casimir operator of $\mathfrak{a}_{(563)}^*$ is given by $K_{563} = x_5^2 + x_6^2$. Using symbolic computing packages, there are $5$ linearly independent and indecomposable solutions to PDE $\{p(x_1,\ldots,x_6),K_{563}\}_o = 0$, which are given by \begin{align*}
    \varphi_1  &= x_3,\quad \varphi_2  = x_5 , \quad \varphi_3  = x_6,  \\
    \varphi_4  & = x_1 x_5 + x_2 x_6 -x_4^2 ,\quad \varphi_5  = x_1 x_6  - x_2 x_5 - 2 x_3 x_4.
\end{align*} Let $\varphi_j(x_1,\ldots,x_6) = A_j$ for all $1 \leq j \leq 5.$ They form the closed quadratic algebra $\tilde{\mathcal{Q}}_{\mathfrak{e} (2)}(2)$ with non-zero Poisson brackets as follows \begin{align}
\nonumber
    \{A_1,A_2\} = & \,  A_3, \text{ }  \{A_1,A_3\} = -A_2, \\
 \{A_2,A_4\} = & \, 2 A_1A_3, \text{ } \{A_3,A_4\} = -2 A_1 A_2. \label{eq:A_2}
\end{align}  By a simple calculation, we can show that $\left\{A_i,\left\{A_j,A_k\right\}\right\} +\left\{A_j,\left\{A_k,A_i\right\}\right\}+\left\{A_k,\left\{A_i,A_j\right\}\right\}= 0$ holds for all $1 \leq i,j,k \leq 5$. Hence $\tilde{\mathcal{Q}}_{\mathfrak{e} (2)}(2)$ form a Poisson quadratic algebra.

We now provide all the functionally independent Casimir functions for $\tilde{\mathcal{Q}}_{\mathfrak{e} (2)}(2).$ Similar to the case in $\mathcal{Q}_{\mathfrak{e} (2)}(2),$ there exist polynomials $t_k^1$ such that \begin{align*}
    K_{\tilde{\mathcal{Q}}_{\mathfrak{e} (2)}(2)}^k =    t_k^1 \left(A_5, A_1^2 + A_4, A_1^2 + A_3^2\right) + \sum_i^k \left(  K_{\tilde{\mathcal{Q}}_{\mathfrak{e} (2)}(2)}^1\right)^i,
\end{align*} where $ K_{\tilde{\mathcal{Q}}_{\mathfrak{e} (2)}(2)}^1 = A_5.$ Hence the functionally independent Casimir functions are \begin{align*}
    K_{\tilde{\mathcal{Q}}_{\mathfrak{e} (2)}(2)}^{1,1} =  A_5, \text{ } \quad   K_{\tilde{\mathcal{Q}}_{\mathfrak{e} (2)}(2)}^{2,1} = A_1^2 +A_4, \text{ } \quad  K_{\tilde{\mathcal{Q}}_{\mathfrak{e} (2)}(2)}^{2,2}  = A_3^2 + A_2^2   . 
\end{align*}

Comparing the result in \cite{MR3988021}, without using the comformal involution on the Hamiltonian, the associated Poisson algebras defined in $\eqref{eq:A_1}$ and $\eqref{eq:A_2}$ are quadratic with more the $3$ generators.

\subsubsection{Subalgebra $\mathfrak{sl}(2)$}

The semisimple Lie algebra $\mathfrak{sl}(2)$ up to conjugacy corresponds to $\mathfrak{a}_{(145)}$ and $\mathfrak{a}_{(246)}$ in $\eqref{eq:3dc}$. Initially, we examine the Casimir operator for $\mathfrak{a}_{(145)}^*$, represented by $K_{145} = x_4^2 - x_1 x_5$. Utilizing symbolic computing packages, there are $5$ indecomposable polynomials such that $\{p(x_1,\ldots,x_6),K_{145}\}_o = 0$, which are given by \begin{align*}
    \varphi_1  &= x_1,\quad \varphi_2  = x_4 , \quad \varphi_3  = x_5,  \\
    \varphi_4  & = x_1 x_6  - x_2 x_5 - 2 x_3 x_4  ,\quad \varphi_5  = x_3^2 + x_2 x_6.
\end{align*} Let $\varphi_j(x_1,\ldots,x_6) = A_j$ for all $1 \leq j \leq 5.$ They form the closed linear Poisson algebra \begin{align}
    \{A_1,A_3\} =  A_1, \text{ }  \{A_1,A_3\} =2A_2, \text{ }   
 \{A_2,A_3\} =  A_3 . \label{eq:A_3}
\end{align} Notice that the center the linear Poisson algebra  is spanned by $A_1$ and $A_4,$ a potential free superintegrable system is then defined by $S_2 =\left\{K_{145},A_1,A_2,A_3\right\}$.

Similarly to the preceding paragraph, the Casimir operator of $\mathfrak{a}_{(246)}^*$ is given by $K_{246} =x_4^2 - x_2 x_6$. Using symbolic computing packages, there are $5$ indecomposable solutions to $\left\{p(x_1,\ldots,x_6),K_{246}\right\}_o = 0$, which are given by \begin{align*}
    \varphi_1  &= x_2,\quad \varphi_2  = x_4 , \quad \varphi_3  = x_6 ,  \\
    \varphi_4  & = x_3^2 + x_1 x_5,\quad \varphi_5  = x_1 x_6  - x_2 x_5 - 2 x_3 x_4.  
\end{align*} Let $\varphi_j(x_1,\ldots,x_6) = A_j$ for all $1 \leq j \leq 5.$ They form the closed linear Poisson algebra \begin{align}
    \{A_1,A_2\} =  A_1, \text{ }  \{A_1,A_3\} =2A_2, \text{ }   
 \{A_2,A_3\} =  A_3 . \label{eq:A_4}
\end{align} 
The quantisation of $\eqref{eq:A_4}$ are \begin{align}
    [\hat{A}_1,\hat{A}_2] =  \hat{A}_1, \text{ }   [\hat{A}_1,\hat{A}_3] =2\hat{A}_2, \text{ }   
  [\hat{A}_2,\hat{A}_3] = \hat{A}_3 . 
\end{align} Therefore, the symmetry algebras corresponding to $\mathfrak{a}_{(145)}$ and $\mathfrak{a}_{(246)}$ can be recognized as $\mathfrak{sl}(2) \oplustilde \mathbb{R}^2$, fulfilling the following short exact sequence \begin{align*}
    0 \rightarrow \mathbb{R}^2 \rightarrow   \mathfrak{sl}(2) \oplustilde \mathbb{R}^2 \rightarrow \mathfrak{sl}(2) \rightarrow 0 
\end{align*} with $\oplustilde$ representing a Lie algebra direct sum. In other words, the polynomial algebra is a central extension of $\mathfrak{sl}(2)$.

\subsubsection{Subalgebra $\mathfrak{su}(2)$}
The Casimir operator of $\mathfrak{su}^* (2 ) $ is given by $K_{\mathfrak{su} (2 )} =x_3^2 +\frac{(x_1 + x_5)^2}{4} + \frac{(x_2 + x_6)^2}{4} $. Using symbolic computing packages, it can be shown that there are $5$ linearly independent and indecomposable polynomial solutions to $\left\{p(x_1,\ldots,x_6),K_{\mathfrak{su} (2 )}\right\}_o = 0$ generating a finite set $\textbf{Q}_3$, which are given by \begin{align*}
    \varphi_1  &= x_1 + x_5,\quad \varphi_2  = x_2 + x_6 , \quad \varphi_3  = x_3 ,  \\
    \varphi_4  & =x_1^2+2 x_4^2+x_5^2-2 x_2 x_6,\quad \varphi_5  = x_1 x_2+2 x_3 x_2+x_5 (2 x_2+x_6).  
\end{align*} Let $\varphi_j(x_1,\ldots,x_6) = A_j$ for all $1 \leq j \leq 5.$ Then $\textbf{Q}_3 = \left\{A_j: 1 \leq j \leq 5\right\},$ which generates the polynomial algebra $\textbf{Alg} \langle \textbf{Q}_3 \rangle$ with the non-trivial bilinear brackets as follows \begin{align}
\nonumber
    \{A_1,A_2\}& =  4A_3, \text{ }  \{A_1,A_3\} =-A_2, \text{ } \{A_1,A_4\} = -4 A_2 A_3, \text{ } \{A_1,A_5\} = 4 A_1 A_3 \\  
 \{A_2,A_3\} & =  A_1, \text{ }   \{A_2,A_4\}  =  -4 A_1A_3, \text{ }  \{A_2,A_5\}  =-4A_2  A_3, \label{eq:38}\\
 \nonumber
 \{A_3,A_4\} & = 2A_1 A_2, \text{ }  \{A_3,A_5\}  =  A_2^2 - A_1^2, \text{ } \{A_4,A_5\}  = 4A_1^2 A_3 + 4 A_2^2 A_3.   
\end{align} For all $1 \leq i,j,k \leq 5,$ we can verify that $\left\{A_i,\left\{A_j,A_k\right\}\right\} +\left\{A_j,\left\{A_k,A_i\right\}\right\}+\left\{A_k,\left\{A_i,A_j\right\}\right\}= 0$ holds by a direct calculation. This guarantee the induced bracket $\{\cdot,\cdot\}$ is Poisson. Hence the relations in $\eqref{eq:38}$ form a Poisson cubic algebra $\mathcal{Q}_{\mathfrak{su}(2)}(3) = \left(\textbf{Alg} \left\langle \textbf{Q}_3 \right\rangle,\{\cdot,\cdot) \right)$ under the Poisson bracket $\{\cdot,\cdot \}$. A direct computation shows that all $A_j$ are functionally independent, we find a superintegrable system $S_3 = \left\{K_{\mathfrak{su}(2)},A_1,A_2,A_3, A_4,A_5\right\}$.  Taking a polynomial \begin{align*}
    K_{\mathcal{Q}_{\mathfrak{su}(2)}(3)}^h = \sum_{i_1 + \cdots + i_5 = h} \Gamma_{i_1,\ldots,i_5} A_1^{i_1} \cdots A_5^{i_5} \in \mathcal{U}_h\left(\mathcal{Q}_{\mathfrak{su}(2)}(3)\right).
\end{align*} We will find all $ K_{\mathcal{Q}_{\mathfrak{su}(2)}(3)}^h$ such that $\left\{ K_{\mathcal{Q}_{\mathfrak{su}(2)}(3)}^h,A_j\right\} = 0$ for all $1 \leq j \leq 5$, and then filter out all the functionally dependent one from those $K_{\mathcal{Q}_{\mathfrak{su}(2)}(3)}^h.$ Since $\mathcal{Z}\left(\mathcal{Q}_{\mathfrak{su}(2)}(3)\right) = \emptyset,$ there is no linear Casimir functions. Inductively, we can conclude that there exist polynomials $h_k$ such that \begin{align*}
     K_{\mathcal{Q}_{\mathfrak{su}(2)}(3)}^k = h_k \left(A_2^2+2 A_3^2+A_4,A_5-A_1 A_2,A_1^2+A_2^2+4 A_3^2\right) 
\end{align*} for any $k \geq 2.$ It follows that there are $3$ functionally independent quadratic Casimir functions \begin{align*}
    K_{\mathcal{Q}_{\mathfrak{su}(2)}(3)}^{2,1} = A_2^2+2 A_3^2+A_4, \text{ }   K_{\mathcal{Q}_{\mathfrak{su}(2)}(3)}^{2,2} =A_5-A_1 A_2,\text{ }   K_{\mathcal{Q}_{\mathfrak{su}(2)}(3)}^{2,3} =A_1^2+A_2^2+4 A_3^2.
\end{align*}

Furthermore, the symmetrization mapping gives a quantisation in $A_j$ for all $j,$ we obtain the following quantum integrals  \begin{align*}
    \hat{A}_1 &= X_1+ X_5, \text{ } \quad \hat{A}_2 = X_2 + X_6, \text{ } \quad  \hat{A}_3 = X_3,\\  
    \hat{A}_4& = X_1^2 + 2 X_4^2 + X_5^2 - 2 \{X_2,X_6\}_a, \\
    \hat{A}_5 &= \{X_1,X_2\}_a + 2\{X_2,X_5\}_a + 2\{X_3,X_4\}_a + \{X_5,X_6\}_a,
\end{align*} which satisfy the non-zero commutators \begin{align}
\nonumber
    [\hat{A}_1,\hat{A}_2]& =  4 \hat{A}_3, \text{ }  [\hat{A}_1, \hat{A}_3] =- \hat{A}_2, \text{ } [\hat{A}_1, \hat{A}_4] = -4  \{\hat{A}_2,  \hat{A}_3\}_a, \text{ } [\hat{A}_1, \hat{A}_5] = 4  \{\hat{A}_1,  \hat{A}_3\}_a \\  
 [\hat{A}_2, \hat{A}_3] & =   \hat{A}_1, \text{ }   [\hat{A}_2, \hat{A}_4]  =  -4  \{\hat{A}_1,\hat{A}_3\}_a, \text{ }  [\hat{A}_2, \hat{A}_5]  =-4 \{ \hat{A}_2,   \hat{A}_3\}_a \text{ } \\
 \nonumber
 [\hat{A}_3, \hat{A}_4] & = 2 \{\hat{A}_1 ,\hat{A}_2\}_a, \text{ }  [\hat{A}_3, \hat{A}_5]  =   \hat{A}_2^2 -  \hat{A}_1^2, \text{ } [\hat{A}_4, \hat{A}_5]  = 4 \hat{A}_1^2  \hat{A}_3 + 4  \{\hat{A}_2^2, \hat{A}_3\}_a - \frac{8}{3} \hat{A}_3,   
\end{align} where $\{\cdot,\cdot\}_a$ defines an anticommutator.  Moreover, by verifying the Jacobi identity of the commutator relations, we see that these elements indeed form an associative algebra.

In Section $\ref{3.8}$, we algebraically construct the polynomial algebras of subalgebra $3$D in $\mathfrak{c}(2)$ related to their Casimir operators. Unlike the polynomial algebra generated from the reduction chain $\mathfrak{g}' \subset \mathfrak{c}(2)$, they form linear, quadratic, or cubic polynomial algebras.

\section{Embedded algebras in quadratic algebras}

\label{5}

In Section $\ref{4},$ we are able to construct finitely-generated algebras in the universal enveloping algebra $\mathcal{U}(\mathfrak{c}(2)).$ 
Six quadratic polynomial Poisson algebras $\mathcal{Q}_j(2)$  are found from the reduction chains $\mathfrak{a}_j \subset \mathfrak{c}(2)$, $j=1,2,\ldots, 6$. In suitable realisations, the generators of the symmetry algebras can be expressed as differential operators.
 
We demonstrate here that the majority of the quadratic Poisson algebras are structured as three-generator quadratic algebras described in \cite{Marquette:2023wxn} with particular parameter selections.

\begin{proposition}
    Let $\mathfrak{c}(2)$ be the $6$-dimensional conformal algebra and $\mathfrak{c}(2)$ be its dual. Then in the realisation defined in $\eqref{eq:real},$ $R(\mathcal{Q}_k(2))$ is a Racah-type algebra for all $k =1,2,5,6$. Moreover, the quadratic algebra $\mathcal{Q}_{\mathfrak{e} (2)}(2)$ defined in Subsection $\ref{4.5.1}$ has a realisation $R(\mathcal{Q}_{\mathfrak{e} (2)}(2)) = \hat{\mathfrak{e}}(2)$, where $\hat{\mathfrak{e}}(2)$ is a central extension of the Euclidean algebra $\mathfrak{e}(2).$
\end{proposition}

\begin{proof}
    We first have a look at the quadratic algebra $\mathcal{Q}_1(2).$ Under the realisation $\eqref{eq:real}$, the integrals that form $\mathcal{Q}_1(2) = \left(\textbf{Alg} \left\langle \textbf{Q}_2^{(1)} \right\rangle,\{\cdot,\cdot\} \right)$ become \begin{align*}
    R(A_1) = p_x, \text{ } R(A_2) = p_y, \text{ } yR(A_3) =   R(A_4) =  R(A_6) = y^2 \left( p_x^2 + p_y^2 \right), \text{ } R(A_5) = 0.
\end{align*} Take $R(A_j) = R_j$ for all $1 \leq j \leq 6.$ Then we have the quadratic algebra with the following non-zero Poisson brackets \begin{align}
\nonumber
    \{R_2,R_3\} =  & \, R_1^2 + R_2^2, \\
    \{R_2,R_4\}  = & \,-2 R_3, \label{eq:5.1} \\
    \nonumber
    \{R_3,R_4\} = & \, -2 R_2 R_4.
\end{align} Notice that $\mathcal{Z}(\mathcal{Q}_1(2)) $ is spanned by $R_1,$ which implies that $R_1$ can be viewed as a new Hamiltonian. Moreover, the Casimir operator of this Racah-type algebra $\eqref{eq:5.1}$ is given by $K_1 = 4R_3^2 + 4 R_2^2 R_4 - 4 R_1^2 R_4.$ For the quadratic algebra $\mathcal{Q}_2(2) =\left(\textbf{Alg} \left\langle \textbf{Q}_2^{(2)} \right\rangle,\{\cdot,\cdot\} \right)$, applying the relation on the generators in $\textbf{Q}_2^{(2)}$ provide \begin{align*}
     R(A_1) = p_x, \text{ } R(A_2) = p_y, \text{ } R(A_3) = x(p_x^2 - p_y^2) , \text{ } R(A_4) = -R(A_6) = x^2 (p_x^2 + p_y^2), \text{ } R(A_5) = 0.
\end{align*}Take $R(A_j) = R_j$ for all $1 \leq j \leq 6.$ The quadratic algebra with non-zero Poisson brackets are given by \begin{align}
\nonumber
    \{R_1,R_3\}  =& \, R_1^2 + R_2^2, \\
    \{R_1,R_4\} = & \,-2 R_3, \label{eq:5.2} \\
    \nonumber
    \{R_3,R_4\} = & \, -2 R_1 R_6.
\end{align} Here $R_2$ generates the center of this algebra. Moreover, the Racah algebra with the relation $\eqref{eq:5.2}$ admits the following Casimir $K_2 = 4 R_3^2 + 4R_1^2 R_4 - 4 R_2^2 R_4.$ Similarly, for the quadratic algebras $\mathcal{Q}_5(2) =\left(\textbf{Alg} \left\langle \textbf{Q}_2^{(5)} \right\rangle,\{\cdot,\cdot\} \right)$ and $\mathcal{Q}_6(2) = \left(\textbf{Alg} \left\langle \textbf{Q}_2^{(6)} \right\rangle,\{\cdot,\cdot\} \right)$, the integrals in $\textbf{Q}_2^{(5)}$ and $\textbf{Q}_2^{(5)}$ under the realisation are \begin{align*}
    R_1 = &\,  (x^2 -y^2) p_x + 2xy p_y, \\
   R_2    =& \, 2 xy p_x +(y^2 - x^2) p_y,  \\ R_3  = & \, -R_5  = y^2(p_x^2 - p_y^2), \\
    R_4   =& \, 0, \text{ }\quad R_6  = (x^2 - y^2)(p_x^2 
 + p_y^2)
 \end{align*} and
 \begin{align*}
     R_1 = & \,(x^2 -y^2) p_x + 2xy p_y,\\
     R_2  =& \, 2 xy p_x +(y^2 - x^2) p_y, \\ R_3 = & \, -R_5 =x^2(p_x^2 + p_y^2), \\
    R_4  =&    \, 0, \text{ }\quad R_6 = (x^2 + y^2)(p_x^2 + p_y^2).
\end{align*}   Thus, we obtain two quadratic algebras described by \begin{align*}
    \{R_2,R_3\} = & \, -2 R_6, \\
    \{R_2,R_6\} =  & \, R_1^2 + R_2^2,  \\
    \{R_3,R_6\} = & \, -2 R_2R_3
\end{align*} with the Casimir $ K_5 =  4 R_6^2 + 4 R_2^2 R_3 - 4 R_1^2 R_3$ and \begin{align*}
    \{R_1,R_3\} = & \,  -2 R_6, \\ 
    \{R_1,R_6\} = & \,  -R_1^2 - R_2^2, \\
    \{R_3,R_6\} = & \, - 2 R_1R_3, 
\end{align*} which admits the Casimir $K_6 = 4 R_6^2 - 4 R_1^2 R_3   +4 R_2^2 R_3$ respectively.

Finally, using $\eqref{eq:real}$, the quadratic algebra $\mathcal{Q}_{\mathfrak{e} (2)}(2)$ defined in $\eqref{eq:A_1}$ has the following realisation     \begin{align*}
    R(A_1) = p_x, \text{ } R(A_2)= p_y , \text{ } - R^2(A_3)=   R(A_4) = -(y p_x - x p_y )^2, \text{ }  R(A_5) = 0.
\end{align*} Let $R(A_k) = R_k$ for all $1 \leq k \leq 3.$ It is clear that $\{R_1,\{R_2,R_3\}\} +\{R_2,\{R_1,R_3\}\}+ \{R_3,\{R_1,R_2\}\}= 0.$ We then have a new algebra generated by $R_1,R_2$ and $R_3$ given by \begin{align*}
 \bar{\mathfrak{e}}(2) = \left\{ R_1,R_2,R_3 :  \{R_1,R_3\} = - R_2, \text{ } \{R_2,R_3\} = R_1\right\}.
\end{align*} It is clear that $\bar{\mathfrak{e}}(2)$ is a central extension of the Euclidean Lie algebra $\mathfrak{e}(2).$ Hence the realisation deforms a quadratic Poisson algebra into a linear Poisson algebra.
 \end{proof}

 \section{Conclusions}
\label{6}

We have systematically studied the problem of commutants, i.e. centralisers, relative to different subalgebras and Casimirs of subalgebras of a Lie algebra. We have demonstrated that such commutants can be used to derive Hamiltonians of algebraic superintegrable systems. This systematic approach is new and provides an alternative framework for the classification of superintegrable systems that do not rely on prior knowledge of differential operator realisations of Hamiltonians and related integrals. 

We have provided procedures for constructing large classes of polynomial algebra structures. As examples, we have presented many types of quadratic algebras whose generators are expressible in terms of polynomials in the enveloping algebra of the 2D conformal algebra.  Interestingly, the commutants of an equivalent class of subalgebras lead to "isomorphic" Poisson polynomial algebras in the sense of a change of basis. From a mathematical point of view, classification of subalgebras of a Lie algebra, especially their maximal Abelian subalgebras, plays an important role in mathematical physics. The results in this paper show that centralisers and their classification with respect to subalgebras are of interest for applications in quantum mechanics and other areas of mathematical physics.

We have obtained Casimir invariants for the quadratic algebras and also presented examples of the reduction from explicit realisations of the polynomials, quadratic algebras, and Casimirs. The connection with the enveloping algebra provides insight into the representation theory of commutants and related special functions.

Note that subalgebra chains also used in other contexts such as dual pairing and Howe duality \footnote{More details on Howe duality can be found in \cite{MR985172,MR986027}, and for duality in the context of representations, see \cite{MR2522486}.} \cite{rowe2012dual,MR2791129}  to study representations related to dynamical symmetry algebras in certain nuclear physics models.  In our approach,  subalgebra chains are applied to find the explicit generators of the centralizer subaglebras, algebraic Hamiltonian, integrals of motion and symmetry algebra. This is achieved through algebraic properties of the symmetric and universal enveloping algebras without using representations.

 Commutants relative to Cartan subalgebras have recently been recognised as useful in generating algebraic Hamiltonians on spheres and in Darboux spaces. So, it is expected that our framework will play an exciting role in the study of superintegrable systems in spaces of non-constant curvatures \cite{MR2337668,MR2226333,MR2023556,MR1878980}.

\section*{CRediT authorship contribution statement}

\textbf{Junze Zhang:} Writing–review and editing, Writing–original draft, Software, Formal analysis. \textbf{Ian Marquette:}  Supervision, Writing–review and editing, Software, Methodology, Investigation, Formal analysis, Conceptualization. \textbf{Yao-Zhong Zhang:} Supervision, Writing–review and editing, Methodology, Investigation, Formal analysis.

\section*{Declaration of competing interest}

The authors declare that they have no known competing financial interests or personal relationships that could have appeared to influence the work reported in this paper.

\section*{Acknowledgment}

IM was supported by the Australian Research Council Future Fellowship FT180100099. YZZ was supported by the Australian Research Council Discovery Project DP190101529. We would like to thank the anonymous referee for the insightful suggestions, which improved the clarity of the manuscript significantly.

\appendix
\section{Appendix}
In this Appendix, we examine various algebraic structures within Lie algebras and utilize those terms to identify all distinct subalgebras of $\mathfrak{so}(3,1).$ Let first us introduce some basic concepts. The following definitions and results can be found in \cite{MR1170508,MR1134846,MR1134849,MR1076112,MR1624621,MR1685593}.

\begin{definition}
Let $\mathfrak{g}$ be a finite-dimensional Lie algebra. A $\textit{maximal}$ $ \textit{Abelian Lie subalgebra}$ $\mathfrak{m}$ (MASA) of $\mathfrak{g}$ is an Abelian subalgebra such that \begin{align*}
    \mathcal{Z}(\mathfrak{m})(\mathfrak{g}) = \{X \in \mathfrak{g}: [X,\mathfrak{m}] = 0\} = \mathfrak{m}.
\end{align*}
\end{definition}

We now give a classification of the conjugacy classes of the maximal Abelian subalgebras of $\mathfrak{c}(2) $ using the results of \cite{MR1170508} and \cite{MR372124}.   
 We briefly recall the results on the classification of MASA of $\mathfrak{so}(p,1)$ given in \cite{gasiorowicz2007quantum}.  As seen later, $\mathfrak{c}(2) \cong \mathfrak{so}(3,1)$.

 Let $K$ be a real non-singular diagonalisable matrix,  and $X \in \mathfrak{m}.$ 
 As can be seen in \cite{MR1170508}, the classification of $SO(p,q)$ conjugacy classes of $\mathfrak{m} \subset \mathfrak{so}(p,q)$ is carried out by specifying the general form of the matrices $X$ and $K$.  Hence, it will be easier to classify the matrix sets $\{\mathfrak{m},K\}.$  We say that two MASAs $\{\mathfrak{m},K\}$ and $\{\mathfrak{m}',K'\}$ of a Lie algebra $\mathfrak{so}(p,q)$ are equivalent if there exists a transformation $G \in \mathrm{GL}_\mathbb{R}(n)$ such that $G\mathfrak{m}G^{-1} =\mathfrak{m}'$ and $G K G^T = K',$  where $\mathfrak{so}(p,q)$ is defined by \begin{align}
    \mathfrak{so}(p,q) = \{Y \in \mathfrak{gl}_\mathbb{R}(n): Y I_{p,q} + I_{p,q}Y^t = 0\}. \label{eq:so}
\end{align} Once a list of representatives of $\mathrm{GL}_\mathbb{R}(n)$ conjugacy classes of the set $\{\mathfrak{m},K\}$ is obtained, a list of $SO(p,q)$ conjugacy classes of MASAs $\mathfrak{m} \subset \mathfrak{so}(p,q)$ can be obtained by specifying $K$. 

\begin{definition}
\label{7.2}
A $\textit{maximal}$ $\textit{Abelian nilpotent subalgebra}$ (MANS)  $\mathfrak{m}$ of  $\mathfrak{g}$ is a MASA that consists entirely of nilpotent elements, which satisfies \begin{align}
    [\mathfrak{m},\mathfrak{m}] = 0, \text{ } \underbrace{[[\mathfrak{g},\mathfrak{m}],\mathfrak{m}],\ldots,\mathfrak{m}]} _\text{  A limited number of commutators} = 0.
\end{align} Moreover, the following abbreviations and terminologies will be used:

1. OD MASA: $\textit{orthogonally decomposable MASA}$. The entire matrix $X \in \mathfrak{m}$ can be represented simultaneously by diagonal block matrices with the same decomposition pattern, that is, $\mathfrak{m} \cong \mathfrak{m}_1 \oplus \mathfrak{m}_2$ or $\mathfrak{m} \cong \mathfrak{m}_3 \oplus \mathfrak{m}_4$. 

2. OID MASA: $\textit{orthogonally indecomposable MASA}$, that is, any MASA which is not OD MASA.
\end{definition}

\begin{proposition}
\cite{MR1624621,MR1685593}
    Let $\mathfrak{so}(p,1)$ be a simple Lie algebra for any $p >0,$ and let $(k,1)$ be the signature of the non-singular diagonalisable matrix $K$, with $k$ positive and one negative eigenvalue of $K$. Then we have the following classification for MASAs of $\mathfrak{so}(p,1)$:
    
\noindent 1. OD MASA. Such a subalgebra has two decomposition patterns, namely: 

$(a)$ $OD(2,1) \oplus (k,1)$ for all $k = 0,1,2,\ldots,p-2,$ where $(k,1)$ is a MANS. 

$(b)$ $ (1,1) \oplus (1,0) \oplus OD(2,0).$

\noindent 2. OID MASA. Such subalgebras are all MANS. 
\end{proposition}

 \begin{remark}
  For any $\textbf{a} \in \mathbb{R}^{p-1},$ a representative list of $O(p,1)$-conjugacy classes of $\mathfrak{so}(p,1)$ is then given by the matrix sets \begin{align*}
    X  = \begin{pmatrix}
        0 & \textbf{a} & 0 \\
        0 & 0 & -\textbf{a}^T\\
        0 & 0 & 0
    \end{pmatrix} \in \mathfrak{m}, \text{ }\quad K = \begin{pmatrix}
        & & 1 \\
        & I_{p-1} & \\
        1 & & 
    \end{pmatrix}.
\end{align*}   \noindent In particular, if $p = 3,$ we can deduce all MASAs of $\mathfrak{so}(3,1)$ up to conjugacy. 
\end{remark}

\vskip 0.4cm
 
 The generators of the Lie algebra $\mathfrak{so}(3,1)$ can be expressed as  \cite{MR372124}
\begin{align*}
    B_1 &=  \begin{pmatrix}
        i & 0 & 0 & 0 \\
        0& -i& 0& 0\\
        0 & 0& -i&  0\\
        0& 0&  0 & i
    \end{pmatrix} , \quad B_2 =  \begin{pmatrix}
        1 & 0 & 0 & 0 \\
        0& -1& 0& 0\\
        0 & 0& 1& 0\\
        0& 0& 0 & -1
    \end{pmatrix}  , \quad B_3 =  \begin{pmatrix}
        0 & 1 & 0 & 0 \\
        0& 0& 0& 0\\
        0 & 0& 0& -1\\
        0& 0& 0 & 0
    \end{pmatrix}  \\
    B_4 & =  \begin{pmatrix}
        0 & i & 0 & 0 \\
        0& 0& 0& 0\\
        0 & 0& 0& i\\
        0& 0& 0 & 0
    \end{pmatrix}  , \quad B_5 =  \begin{pmatrix}
        0 & 0 & 0 & 0 \\
        1& 0& 0& 0\\
       0 & 0& 0& 0\\
        0& 0& -1 & 0
    \end{pmatrix}  ,\quad B_6 = \begin{pmatrix}
        0 & 0 & 0 & 0 \\
        i& 0& 0& 0\\
        0 & 0& 0& 0\\
        0& 0& i & 0
    \end{pmatrix} .
\end{align*}  Notice that the basis chosen here is different from the one conventionally used in physics.

From the above proposition and the classification in \cite{MR372124}, we can provide the classification of all subalgebras of $\mathfrak{so}(3,1)$. With such a classification and the isomorphism between $\mathfrak{so}(3,1)$ and $\mathfrak{c}(2)$,  all subalgebras of $\mathfrak{c}(2)$ can be identified.

\begin{corollary}
\label{clas}
There are only two types of MASAs for $\mathfrak{so}(3,1)$: \begin{align*}
   OD(2,1) \oplus (1,1)  : \quad  X  & = \begin{pmatrix}
        0 & a_1 & 0 & 0 \\
        -a_1 & 0 & 0 & 0\\
        0 & 0 & a_6 & 0 \\
        0 & 0 & 0 &- a_6 \\
    \end{pmatrix}, \text{ }\quad K = \begin{pmatrix}
      1  & 0& 0 & 0 \\
       0 &1 &0 & 0 \\
        0 & 0 & 0 & 1 \\
        0 & 0 & 1 & 0 
    \end{pmatrix} \\
   MANS  : \quad X  & = \begin{pmatrix}
        0 & a_1 & a_2 & 0 \\
        0 & 0 & 0 & -a_1\\
        0 & 0 & 0 & -a_2 \\
        0 & 0 & 0 & 0 \\
    \end{pmatrix}, \text{ } \quad K = \begin{pmatrix}
        & & 1 \\
        & I_2 & \\
        1 & & 
    \end{pmatrix}.
\end{align*} Other subalgebras of $\mathfrak{so}(3,1)$, unique up to a conjugation, are given by \begin{align*}
     F_1 : \{B_1,B_2,B_3,B_4\}, &\qquad F_8 : \{B_1,B_2\}, \\
     F_2 :   \{B_1,B_3 -  B_5,B_4 + B_6\}, &\qquad  F_9 : \{B_3,B_4\}, \\
     F_3: \{B_1,B_3 + B_5,B_4 - B_6\}, &\qquad  F_{10} : \{\cos \theta B_1 + \sin \theta B_2\}, \\
     F_4 : \{\cos \theta B_1 + \sin \theta B_2,B_3,B_4\}, &\qquad  F_{11} : \{ B_1\}, \\
     F_5 : \{B_1,B_3,B_4\}, &\qquad  F_{12} : \{B_2\}, \\
       F_6 : \{B_2,B_3,B_4\}, &\qquad  F_{13} : \{B_3\}, \\
       F_7 : \{B_2,B_3\} ,  
\end{align*}  
with the commutator relations  \begin{center}
      \begin{tabular}{|c|c c c c c c| }
      \hline
           & $B_1$ & $B_2$ & $B_3$ & $B_4$ & $B_5$ & $B_6$   \\
           \hline
    $B_1$       & 0 & 0 & $2B_4$ & $-2B_3$& $-2B_6$ & $2B_5$  \\
     $B_2$       & 0 &  0 & $2B_3$ &$2B_4$ &$-2B_5$ & $-2B_6$  \\
      $B_3$        & $ -2B_4$ &$-2B_3$ & 0 &0 &$ B_2$ & $B_1$ \\
       $B_4$        & $2B_3$ & $-2B_4$& 0& 0&$B_1$ & $-B_2$ \\
      $B_5$         & $2B_6$ & $2B_5$ & $-B_2$ & $-B_1$ &0 & 0 \\
        $B_6$     & $-2B_5$ & $2B_6$ & $-B_1$ & $B_2$ & 0&  0\\ 
        \hline
      \end{tabular}

      \quad

    \textbf{Table 2}\label{0}  
  \end{center} Here $0 < \theta < \pi, \theta \neq \frac{\pi}{2}.$
 
\end{corollary}

\begin{remark}
    \label{clasr}
The one-dimensional subalgebra $F_{10}$ gives an infinitesimal generator for each value of $\theta$. Other one-dimensional subalgebras $F_{11}$, $ F_{12} $ and $F_{13} $ can be realised into one-dimensional Lie algebras $ \mathfrak{o}(2)$,  $ \mathfrak{so}(1,1)$ and  $\mathfrak{e}(1)$ respectively. For the two-dimensional subalgebras of $\mathfrak{so}(3,1)$, we observe that $F_7$ is isomorphic to the Borel subalgebra of $\mathfrak{sl}(2)$, that is, $F_7 \cong \mathfrak{b}_0$. The other two Abelian subalgebras, $F_8$ and $F_9$, correspond to OD MASAs and MANS, respectively.  Finally, three-dimensional subalgebras are given by $F_2 \cong \mathfrak{so}(3)$, $F_3 \cong \mathfrak{su}(1,1)$, and $F_5$ and $F_6$ are isomorphic to $\mathfrak{sl}(2).$
\end{remark}

\bibliographystyle{unsrt}
\bibliography{bibliography.bib}

\end{document}